\documentclass[a4paper,10pt]{article}
\usepackage{geometry}
\usepackage[utf8]{inputenc}
\usepackage[T1]{fontenc}
\usepackage{graphicx}
\usepackage{amsmath, amssymb, amstext, mathtools}
\usepackage{amsthm}
\usepackage{tikz}
\usepackage{pgfplots}
\usepackage{float}
\usepackage{hyperref}
\usepackage{paralist}
\usepackage[titlenumbered, linesnumbered, ruled]{algorithm2e}

\usetikzlibrary{patterns,arrows,decorations.pathreplacing,decorations.pathmorphing,calc}

\pgfplotsset{compat=1.9}

\newtheorem{theorem}{Theorem}[section]
\newtheorem{lemma}[theorem]{Lemma}
\newtheorem{corollary}[theorem]{Corollary}

\newtheorem{observation}[theorem]{Observation}
\theoremstyle{definition}
\newtheorem{definition}[theorem]{Definition}
\theoremstyle{remark}
\newtheorem{remark}[theorem]{Remark}

\DeclareMathOperator{\Aut}{Aut}
\DeclareMathOperator{\Out}{Out}
\DeclareMathOperator{\Sym}{Sym}
\DeclareMathOperator{\Alt}{Alt}
\DeclareMathOperator{\Cl}{Cl}

\DeclareMathOperator{\Soc}{Soc}

\DeclareMathOperator{\Iso}{Iso}

\DeclareMathOperator{\Aff}{Aff}

\DeclareMathOperator{\im}{im}

\DeclareMathOperator{\dist}{dist}

\DeclareMathOperator{\ourgamma}{\widehat{\Gamma}}

\DeclareMathOperator{\AGL}{AGL}
\DeclareMathOperator{\GL}{GL}
\DeclareMathOperator{\SL}{SL}
\DeclareMathOperator{\SU}{SU}
\DeclareMathOperator{\Sp}{Sp}

\numberwithin{equation}{section}
\numberwithin{figure}{section}

\newcommand{\angles}[1]{\left\langle#1\right\rangle}

\tikzstyle{vertex}=[circle,fill=white,draw=black]

\newcounter{claimcounter}
\newenvironment{claim}[1][]{
  \renewcommand{\proof}{\smallskip\par\noindent\textit{Proof. }}
  \medskip\par\noindent%
  \ifthenelse{\equal{#1}{}}{%
    \setcounter{claimcounter}{0}\refstepcounter{claimcounter}\textit{Claim~\arabic{claimcounter}.}
  }{%
    \ifthenelse{\equal{#1}{resume}}{%
      \refstepcounter{claimcounter}\textit{Claim~\arabic{claimcounter}.}
    }{%
      \textit{Claim~#1.}
    }
  }
}{
  \par\medskip
}
\newcommand{\uend}{\hfill$\lrcorner$}

\newcommand{\polylog}[1]{\operatorname{polylog}(#1)}

\title{A Faster Isomorphism Test for Graphs of Small Degree}
\author{
Martin Grohe\\
RWTH Aachen University\\
\texttt{grohe@informatik.rwth-aachen.de}
\and
Daniel Neuen\\
RWTH Aachen University\\
\texttt{neuen@informatik.rwth-aachen.de}
\and
Pascal Schweitzer\\
TU Kaiserslautern\\
\texttt{schweitzer@cs.uni-kl.de}
}
\date{}

\begin{document}

\maketitle

\begin{abstract}
 In a recent breakthrough, Babai (STOC 2016) gave a quasipolynomial time
 graph isomorphism test. In this work, we give an improved
 isomorphism test for graphs of small degree: our algorithms runs in time 
 $n^{\mathcal{O}((\log d)^{c})}$, where $n$ is the number of vertices of
 the input graphs, $d$ is the maximum degree of the input graphs, and
 $c$ is an absolute constant.  
 The best previous isomorphism test for graphs of maximum degree $d$ due to
 Babai, Kantor and Luks (FOCS 1983) runs in time
 $n^{\mathcal{O}(d/ \log d)}$.  
\end{abstract}

\section{Introduction}

Luks's polynomial time isomorphism test for graphs of bounded degree \cite{luks82} is one of the cornerstones of the algorithmic theory of graph isomorphism.
With a slight improvement given later~\cite{BKL83}, it tests in time
$n^{\mathcal{O}(d/\log d)}$ 
whether two $n$-vertex graphs of maximum degree $d$ are isomorphic.
Over the past decades Luks's algorithm and its algorithmic framework have been used as a building block for many isomorphism algorithms
(see e.g.\
\cite{BKL83,BL83,GM15,KS17,Luks91,Ponomarenko91,Seress03}). More importantly, 
it also forms the basis for Babai's recent isomorphism test for general graphs \cite{Babai15-full, Babai16} which runs in quasipolynomial time (i.e.,\ the running time is bounded by $n^{\polylog{n}}$).
Indeed, Babai's algorithm follows Luks's algorithm, but attacks the obstacle cases for which the recursion performed by Luks's framework does not lead to the desired running time.
Graphs whose maximum degree $d$ is at most polylogarithmic in the
number $n$ of vertices are not a critical case for Babai's algorithm,
because for such graphs no large alternating or symmetric groups appear
as factors of the automorphism group, and therefore the running time
of Babai's algorithm on the class of all these graphs is still quasipolynomial. 
Hence graphs of polylogarithmic maximum degree form one of the obstacle cases
towards improving Babai's algorithm. This alone is a strong motivation
for trying to improve Luks's algorithm. In view of Babai's
quasipolynomial time algorithm, it is natural to ask whether there is
an $n^{\polylog{d}}$-isomorphism test for graphs of maximum degree
$d$. In this paper we answer this question affirmatively.

\begin{theorem}\label{thm:main-result-degree-d}
 The Graph Isomorphism Problem for graphs of maximum degree $d$ can be solved in time $n^{\mathcal{O}((\log d)^c)}$, for an absolute constant $c$.
\end{theorem}

To prove the result we follow the standard route of considering the \emph{String Isomorphism
Problem}, which is an abstraction of the Graph Isomorphism Problem
that has been introduced by Luks in order to facilitate a recursive
isomorphism test based on the structure of the permutation groups
involved \cite{BL83,luks82}. Here a \emph{string} is simply a mapping
$\mathfrak x:\Omega\to\Sigma$, where the \emph{domain} $\Omega$ and
\emph{alphabet} $\Sigma$ are just
finite sets. Given two strings
$\mathfrak x,\mathfrak y:\Omega\to\Sigma$ and a permutation group
$G\le\Sym(\Omega)$, the objective of the String Isomorphism Problem is
to compute the set $\Iso_G(\mathfrak x,\mathfrak y)$ of all
\emph{$G$-isomorphisms} from $\mathfrak x$ to $\mathfrak y$, that is,
all permutations $g\in G$ mapping $\mathfrak x$ to $\mathfrak y$. 
We study the String Isomorphism Problem for groups
$G$ in the class $\ourgamma_d$ of groups all of whose composition
factors are isomorphic to subgroups of $S_d$, the symmetric group
acting on $d$ points. 
Luks introduced this class because he observed that, after fixing a
single vertex, the automorphism group of a connected graph of maximum degree $d$
is in $\ourgamma_d$\footnote{In \cite{luks82}, the class
$\ourgamma_d$ is denoted by $\Gamma_d$. However, in the more recent
literature $\Gamma_d$ typically refers to a larger class of groups
\cite{BCP82} (see Subsection~\ref{subsubsec:restr:comp:fac}).}.
Our main technical result, Theorem~\ref{thm:main-result-gamma-d},
states that we can solve the String
Isomorphism Problem for groups $G\in\ourgamma_d$ in time
$n^{\polylog{d}}$, where $n=|\Omega|$ is the length of the input
strings. This implies Theorem~\ref{thm:main-result-degree-d} (as outlined in Section~\ref{sec:applications}).

To prove this result, we introduce the new concept of an
\emph{almost~$d$-ary sequence} 
of invariant
partitions. More precisely, we exploit for the group~$G$ 
a sequence $\{\Omega\} = \mathfrak{B}_0 \succ \dots \succ \mathfrak{B}_m =
\{\{\alpha\} \mid \alpha \in \Omega\}$ of $G$-invariant partitions
$\mathfrak B_i$ of $\Omega$, where 
$\mathfrak B_{i-1}\succ\mathfrak B_i$ means that $\mathfrak B_i$
refines $\mathfrak B_{i-1}$.
For this sequence we require that for all $i$ 
the induced group of permutations of the subclasses in~$\mathfrak B_i$
of a given class in~$\mathfrak B_{i-1}$ is permutationally equivalent to a subgroup of
the symmetric group~$S_d$ or semi-regular (i.e., only the identity has fixed points). 
Our algorithm that exploits such a sequence is heavily based on
techniques introduced by Babai for his quasipolynomial time
isomorphism test. We even use Babai's algorithm as a black box in one
case. One of our technical contributions is an adaptation of Babai's
Unaffected Stabilizers Theorem \cite[Theorem 6]{Babai16} to groups
constrained by an almost $d$-ary sequence of invariant
partitions. In~\cite{Babai16}, the Unaffected Stabilizers Theorem lays
the groundwork for the group theoretic algorithms (the Local
Certificates routine), and it plays a similar role here. However, we
need a more refined running time analysis.
Based on this we
can then adapt the Local Certificates routine to our setting. 

However, not every group in~$\ourgamma_d$ has such an almost~$d$-ary
sequence required by our technique. We remedy this by changing the
operation of the group while preserving string isomorphisms. The
structural and algorithmic results enabling such a change of operation
form the second technical contribution of our work.  For this we
employ some heavy group theoretic results. First, applying the
classification of finite simple groups via the O'Nan-Scott Theorem and
several other group theoretic characterizations, we obtain a structure
theorem for primitive permutation groups in~$\ourgamma_d$ showing that
they are either small (of size at most $n^{\polylog{d}}$) or have a
specific structure. More precisely, large primitive groups
in~$\ourgamma_d$ are composed, in a well defined manner, of Johnson groups
(i.e.\ symmetric/alternating groups with an induced action on $t$-element subsets of the standard domain).
Second, to construct the almost~$d$-ary sequence of partitions,
we exploit the existence of these Johnson schemes and introduce subset
lattices which are unfolded yielding the desired group operation.

With Luks's framework being used as a subroutine in various other
algorithms, one can ask for the impact of the improved running time in
such contexts. 
As a first, simple application we obtain an improved 
isomorphism test for relational structures
(Theorem~\ref{thm:relational}) and hypergraphs (Corollary~\ref{cor:hypergraph}). 
A deeper application is an improved fixed-parameter tractable algorithm
for graph isomorphism of graphs
parameterized by tree width~\cite{GroheNSW18}, which substantially improves the
algorithm from~\cite{LPPS14}. 

\paragraph{Outline} 
Section~\ref{sec:characterization-primitive} is
concerned with the structure of primitive $\ourgamma_d$ groups; it
culminates in Theorem~\ref{thm:first-main-theorem} with a structural
description. In Section~\ref{sec:almost:d:ary} we describe how to
algorithmically change the operation of a group in $\ourgamma_d$ to
force the existence of an almost $d$-ary sequence of
invariant partitions
$\{\Omega\} = \mathfrak{B}_0 \succ \dots \succ \mathfrak{B}_m =
\{\{\alpha\} \mid \alpha \in \Omega\}$
without changing string isomorphisms. In
Sections~\ref{sec:affected:orbits} and~\ref{sec:local-certificates} we
extend Babai's structural group theoretic results to our
situation. Finally, in Section~\ref{sec:algorithm} we prove our main
algorithmic results,
Theorems~\ref{thm:string-isomorphism-almost-d-ary} and
\ref{thm:main-result-gamma-d}. Applications of these results, among
them Theorem~\ref{thm:main-result-degree-d}, are presented in
Section~\ref{sec:applications}.

\section{Preliminaries}

\subsection{Graphs and other structures}

A \emph{graph} is a pair $\Gamma=(V,E)$ with vertex set $V = V(\Gamma)$ and edge relation $E = E(\Gamma)$.
In this paper all graphs are finite and simple, i.e.\ there are no loops or multiple edges. The \emph{neighborhood} of~$v\in V(\Gamma)$ is denoted~$N(v)$.
A \emph{path} of length $k$ is a sequence $v_0,\dots,v_k$ of distinct vertices such that $(v_{i-1},v_i) \in E(\Gamma)$ for all $i \in [k]$ (where $[k] := \{1,\dots,k\}$).
The \emph{distance} between two vertices $v,w \in V(\Gamma)$, denoted by $\dist(v,w)$, is the length of the shortest path from $v$ to $w$.

An \emph{isomorphism} from a graph $\Gamma_1$ to another graph
$\Gamma_2$ is a bijective mapping $\varphi\colon V(\Gamma_1)
\rightarrow V(\Gamma_2)$ which preserves the edge relation, that is, $(v,w) \in E(\Gamma_1)$ if and only if $(\varphi(v),\varphi(w)) \in E(\Gamma_2)$ for all~$v,w \in V(\Gamma_1)$.
Two graphs $\Gamma_1$ and $\Gamma_2$ are \emph{isomorphic} ($\Gamma_1 \cong \Gamma_2$) if there is an isomorphism from~$\Gamma_1$ to~$\Gamma_2$.
An \emph{automorphism} of a graph $\Gamma$ is an isomorphism from~$\Gamma$ to itself. By $\Aut(\Gamma)$ we denote the group of automorphisms of $\Gamma$.
The \emph{Graph Isomorphism Problem} asks, given two (undirected) graphs $\Gamma_1$ and $\Gamma_2$, whether they are isomorphic.

A \emph{$t$-ary relational structure} is a tuple $\mathfrak{A} = (D,R_1,\dots,R_k)$ with domain $D$ and $t$-ary relations $R_i \subseteq D^{t}$ for $i \in [k]$.
An \emph{isomorphism} from a structure $\mathfrak{A}_1 = (D_1,R_1,\dots,R_k)$ to another structure $\mathfrak{A}_2 = (D_2,S_1,\dots,S_k)$ is a bijective mapping $\varphi\colon D_1 \rightarrow D_2$ such that $(v_1,\dots,v_t) \in R_i$ if and only if $(\varphi(v_1),\dots,\varphi(v_t)) \in S_i$ for all~$v_1,\dots,v_t \in D_1$ and $i \in [k]$.
As before, $\Aut(\mathfrak{A})$ denotes the automorphism group of $\mathfrak{A}$.

Let $\mathfrak{B}_1,\mathfrak{B}_2$ be two partitions of the same set $\Omega$.
We say $\mathfrak{B}_1$ \emph{refines} $\mathfrak{B}_2$, denoted by $\mathfrak{B}_1 \preceq \mathfrak{B}_2$, if for every $B_1 \in \mathfrak{B}_1$ there is some $B_2 \in \mathfrak{B}_2$ such that $B_1 \subseteq B_2$.
If additionally $\mathfrak{B}_1$ and $\mathfrak{B}_2$ are distinct we say $\mathfrak{B}_1$ \emph{strictly refines} $\mathfrak{B}_2$ ($\mathfrak{B}_1 \prec \mathfrak{B}_2$).
The \emph{index} of $\mathfrak B_1$ in $\mathfrak B_2$
is $|\mathfrak{B}_2 : \mathfrak{B}_1| \coloneqq \max_{B_2 \in \mathfrak{B}_2} |\{B_1 \in \mathfrak{B}_1 \mid B_1 \subseteq B_2\}|$.
A partition $\mathfrak{B}$ (of the set $\Omega$) is an \emph{equipartition} if all elements $B \in \mathfrak{B}$ have the same size.
For $S \subseteq \Omega$ we define the \emph{induced partition} $\mathfrak{B}[S] = \{B \cap S \mid B \in \mathfrak{B} \text{ such that } B \cap S \neq \emptyset\}$.
Note that $\mathfrak{B}[S]$ forms a partition of the set $S$.

For a set $M$ and a natural number $t \leq |M|$ we denote by $\binom{M}{t}$ the set of all $t$-element subsets of $M$, that is, $\binom{M}{t} = \{X \subseteq M \mid |X| = t\}$.
Note that the number of elements in $\binom{M}{t}$ is exactly $\binom{|M|}{t}$.
Moreover, $\binom{M}{\leq t}$ denotes the set of all subsets of $M$ of cardinality at most $t$. 

\subsection{Group Theory}

In this section we introduce the group theoretic notions required in this work.
For a general background on group theory we refer to \cite{Rotman99} whereas background on permutation groups can be found in \cite{DM96}.

\subsubsection{Permutation groups}

A \emph{permutation group} acting on a set $\Omega$ is a subgroup $G \leq \Sym(\Omega)$ of the symmetric group.
The size of the permutation domain $\Omega$ is called the \emph{degree} of $G$ and, throughout this work, is denoted by $n = |\Omega|$.
If $\Omega = [n]$ then we also write $S_n$ instead of $\Sym(\Omega)$.
For $g \in G$ and $\alpha \in \Omega$ we denote by $\alpha^{g}$ the image of $\alpha$ under the permutation $g$.
The set $\alpha^{G} = \{\alpha^{g} \mid g \in G\}$ is the \emph{orbit} of $\alpha$.
The group $G$ is \emph{transitive} if $\alpha^{G} = \Omega$ for some (and therefore every) $\alpha \in \Omega$.

For $\alpha \in \Omega$ the group $G_\alpha = \{g \in G \mid \alpha^{g} = \alpha\} \leq G$ is the \emph{stabilizer} of $\alpha$ in $G$.
The group $G$ is \emph{semi-regular} if $G_\alpha = \{1\}$ for all $\alpha \in \Omega$.
If additionally $G$ is transitive then the group $G$ is called \emph{regular}.
For $\Delta \subseteq \Omega$ and $g \in G$ let $\Delta^{g} = \{\alpha^{g} \mid \alpha \in \Delta\}$.
The \emph{pointwise stabilizer} of $\Delta$ is the subgroup $G_{(\Delta)} = \{g \in G \mid\forall \alpha \in \Delta\colon \alpha^{g}= \alpha \}$.
The \emph{setwise stabilizer} of $\Delta$ is the subgroup $G_{\Delta} = \{g \in G \mid \Delta^{g}= \Delta\}$.
Observe that $G_{(\Delta)} \leq G_{\Delta}$.

Let $G \leq \Sym(\Omega)$ be a transitive group.
A \emph{block} of $G$ is a nonempty subset $B \subseteq \Omega$ such that $B^{g} = B$ or $B^{g} \cap B = \emptyset$ for all $g \in G$.
The trivial blocks are $\Omega$ and the singletons $\{\alpha\}$ for
$\alpha \in \Omega$.
The group $G$ is said to be \emph{primitive} if there are no non-trivial blocks.
If $G$ is not primitive it is called \emph{imprimitive}.
If $B \subseteq \Omega$ is a block of $G$ then $\mathfrak{B} = \{B^{g} \mid g \in G\}$ forms a \emph{block system} of $G$.
Note that $\mathfrak{B}$ is an equipartition of $\Omega$.
The group $G_{(\mathfrak{B})} = \{g \in G \mid \forall B \in \mathfrak{B}\colon B^{g} = B\}$ denotes the subgroup stabilizing each block $B \in \mathfrak{B}$ setwise.
Observe that $G_{(\mathfrak{B})}$ is a normal subgroup of~$G$.
We denote by $G^{\mathfrak{B}} \leq \Sym(\mathfrak{B})$ the natural action of $G$ on the block system $\mathfrak{B}$.
More generally, if $A$ is a set of objects on which $G$ acts naturally, we denote by $G^{A} \leq \Sym(A)$ the action of $G$ on the set $A$.
A block system $\mathfrak{B}$ is \emph{minimal} if there is no non-trivial block system $\mathfrak{B}'$ such that $\mathfrak{B} \prec \mathfrak{B}'$.
Note that a block system $\mathfrak{B}$ is minimal if and only if $G^{\mathfrak{B}}$ is primitive.

A class of primitive groups that plays an important role in this work is the class of \emph{Johnson groups}, alternating and symmetric groups with their actions on $t$-element subsets of the standard domain.
For $m \in \mathbb{N}$ we denote by $A_m$ the alternating group acting on the set $[m]$.
For $t \leq \frac{m}{2}$ let $A_m^{(t)}$ be the action of $A_m$ on the set of $t$-element subsets of $[m]$.
Similarly, $S_m^{(t)}$ denotes the action of $S_m$ on the set of $t$-element subsets of $[m]$.

Let $G \leq \Sym(\Omega)$ and $G' \leq \Sym(\Omega')$.
A \emph{homomorphism} is a mapping $\varphi\colon G \rightarrow G'$ such that $\varphi(g)\varphi(h) = \varphi(gh)$ for all $g,h \in G$.
For $g \in G$ we denote by $g^{\varphi}$ the $\varphi$-image of $g$.
Similarly, for $H \leq G$ we denote by $H^{\varphi}$ the $\varphi$-image of $H$ (note that $H^{\varphi}$ is a subgroup of $G'$).

A \emph{permutational isomorphism} from $G$ to $G'$ is a bijective mapping $f:\Omega\to\Omega'$
such that $G' = \{f^{-1}gf \mid g \in G\}$ where $f^{-1}gf\colon \Omega' \rightarrow \Omega'$ is the unique map mapping~$f(\alpha)$ to~$f(\alpha^{g})$ for all~$\alpha\in \Omega'$.
If there is a permutational isomorphism from $G$ to $G'$, we call $G$ and $G'$ \emph{permutationally equivalent}.
A \emph{permutational automorphism} of $G$ is a permutational isomorphism from $G$ to itself.

\subsubsection{Algorithms for permutation groups}

We review some basic facts about algorithms for permutation groups.
For detailed information we refer to \cite{Seress03}.

In order to perform computational tasks for permutation groups efficiently the groups are represented by generating sets of small size.
Indeed, most algorithms are based on so called strong generating sets,
which can be chosen of size quadratic in the size of the permutation domain of the group and can be computed in polynomial time given an arbitrary generating set.

\begin{theorem}[cf.\ \cite{Seress03}] 
 \label{thm:permutation-group-algorithm-library}
 Let $G \leq \Sym(\Omega)$ and let $S$ be a generating set for $G$.
 Then the following tasks can be performed in time polynomial in $n$ and $|S|$:
 \begin{enumerate}
  \item compute the order of $G$,
  \item given $g \in \Sym(\Omega)$, test whether $g \in G$,
  \item compute the orbits of $G$,
  \item given $\Delta \subseteq \Omega$, compute a generating set for $G_{(\Delta)}$, and
  \item compute a minimal block system for $G$.
 \end{enumerate}
 For a second group $G' \leq \Sym(\Omega')$ with domain size $n' = |\Omega'|$, the following tasks can be solved in time polynomial in $n$, $n'$ and $|S|$:
 \begin{enumerate}
  \item[6.] given a homomorphism $\varphi\colon G \rightarrow G'$ (given as a list of images for $g \in S$),
  \begin{enumerate}
   \item compute a generating set for $\ker(\varphi) = \{g \in G \mid \varphi(g) = 1\}$, and
   \item given $g'\in G'$, compute an element $g \in G$ such that
     $\varphi(g) = g'$ (if it exists).
  \end{enumerate}
 \end{enumerate}
\end{theorem}

\subsubsection{Groups with restricted composition factors}\label{subsubsec:restr:comp:fac}

In this work we shall be interested in a particular subclass of permutation groups, namely groups with restricted composition factors.
Let $G$ be a group.
A \emph{subnormal series} is a sequence of subgroups $G = G_0 \trianglerighteq G_1 \trianglerighteq \dots \trianglerighteq G_k = \{1\}$.
The length of the series is $k$ and the groups $G_{i-1} / G_{i}$ are the factor groups of the series, $i \in [k]$.
A \emph{composition series} is a strictly decreasing subnormal series of maximal length. 
For every finite group $G$ all composition series have the same family of factor groups considered as a multi-set (cf.\ \cite{Rotman99}).
A \emph{composition factor} of a finite group $G$ is a factor group of a composition series of $G$.

\begin{lemma}[\cite{BCP82}, Lemma 2.2]
 \label{la:size-gamma-d}
 Suppose $d \geq 6$.
 Let $G$ be a permutation group of degree $n$ such that $G$ has no composition factor isomorphic to an alternating group $A_k$ of degree $k > d$.
 Then $|G| \leq d^{n-1}$.
\end{lemma}

\begin{definition}
 For $d \geq 2$ let $\ourgamma_d$ denote the class of all groups $G$ for which every composition factor of $G$ is isomorphic to a subgroup of $S_d$.
\end{definition}

We want to stress the fact that there are two similar classes of
groups that have been used in the literature both typically denoted by~$\Gamma_d$.
One of these is the class we define as~$\ourgamma_d$ introduced by Luks~\cite{luks82} while the other one used in~\cite{BCP82} in particular allows composition factors that are simple groups of Lie type of bounded dimension.

\begin{lemma}[Luks \cite{luks82}]
 \label{la:gamma-d-closure}
 Let $G \in \ourgamma_d$. Then
 \begin{enumerate}
  \item $H \in \ourgamma_d$ for every subgroup $H \leq G$, and
  \item $G^{\varphi} \in \ourgamma_d$ for every homomorphism $\varphi\colon G \rightarrow H$.
 \end{enumerate}
\end{lemma}

\subsubsection{String Isomorphism and Luks's algorithm}

In the following we give an outline of Luks's algorithm \cite{luks82}.
Our description of the algorithm as well as the notation mainly follows \cite{Babai16}.

Let $\mathfrak{x},\mathfrak{y}\colon \Omega \rightarrow \Sigma$ be two strings over a finite alphabet $\Sigma$ and let $G \leq \Sym(\Omega)$ be a group.
For $\sigma \in \Sym(\Omega)$ the string $\mathfrak{x}^{\sigma}$ is defined by \[\mathfrak{x}^{\sigma}(\alpha) = \mathfrak{x}(\alpha^{\sigma^{-1}})\] for all $\alpha \in \Omega$.
A permutation $\sigma \in \Sym(\Omega)$ is a \emph{$G$-isomorphism} from $\mathfrak{x}$ to $\mathfrak{y}$ if $\sigma \in G$ and $\mathfrak{x}^{\sigma} = \mathfrak{y}$.
The \emph{String Isomorphism Problem} asks, given $\mathfrak{x},\mathfrak{y}\colon \Omega \rightarrow \Sigma$ and a group $G \leq \Sym(\Omega)$ given as a set of generators, whether there is a $G$-isomorphism from $\mathfrak{x}$ to $\mathfrak{y}$.
The set of $G$-isomorphisms is denoted by $\Iso_G(\mathfrak{x},\mathfrak{y}) := \{g \in G \mid \mathfrak{x}^{g} = \mathfrak{y}\}$.

More generally, for $K \subseteq \Sym(\Omega)$ and $W \subseteq \Omega$ we define
\begin{equation}
 \Iso_K^{W}(\mathfrak{x},\mathfrak{y}) = \{g \in K \mid \forall \alpha \in W\colon \mathfrak{x}(\alpha) = \mathfrak{y}(\alpha^{g})\}.
\end{equation}
In this work $K = Gg$ will always be a coset where $G \leq \Sym(\Omega)$ and $g \in \Sym(\Omega)$ and the set $W$ will be $G$-invariant.
In this case $\Iso_K^{W}(\mathfrak{x},\mathfrak{y})$ is either empty or a coset of the group $\Aut_G^{W}(\mathfrak{x}) := \Iso_G^{W}(\mathfrak{x},\mathfrak{x})$, that is, $\Iso_K^{W}(\mathfrak{x},\mathfrak{y}) = \Aut_G^{W}(\mathfrak{x})\sigma$ where $\sigma \in \Iso_K^{W}(\mathfrak{x},\mathfrak{y})$ is arbitrary.
Hence, the set $\Iso_K^{W}(\mathfrak{x},\mathfrak{y})$ can be represented by a generating set for $\Aut_G^{W}(\mathfrak{x})$ and an element $\sigma$.
Moreover, using the identity
\begin{equation}
 \label{eq:string-align}
 \Iso_{Gg}^{W}(\mathfrak{x},\mathfrak{y}) = \Iso_G^{W}(\mathfrak{x},\mathfrak{y}^{g^{-1}})g,
\end{equation}
it is actually possible to restrict ourselves to the case where $K$ is a group.

We now describe the two main
recursive steps used in Luks's algorithm \cite{luks82}.
First suppose $G \leq \Sym(\Omega)$ is not transitive and let $\Omega_1,\dots,\Omega_s$ be the orbits of $G$.
Then the strings are processed orbit by orbit as described in Algorithm \ref{alg:orbit-by-orbit}.

\begin{algorithm}
 \caption{Orbit-by-Orbit processing}
 \label{alg:orbit-by-orbit}
 \DontPrintSemicolon
 $K := G$\;
 \For{$i=1,\dots,s$}{
  $K := \Iso_K^{\Omega_i}(\mathfrak{x},\mathfrak{y})$\;
 }
 \Return $K$\;
\end{algorithm}

Note that the set $\Iso_K^{\Omega_i}(\mathfrak{x},\mathfrak{y})$ can be computed making one call to String Isomorphism over domain size $n_i = |\Omega_i|$.
Indeed, using Equation (\ref{eq:string-align}), it can be assumed that $K \leq \Sym(\Omega)$ is a group and $\Omega_i$ is $K$-invariant.
Then
\[\Iso_K^{\Omega_i}(\mathfrak{x},\mathfrak{y}) = \left\{k \in K \mid k^{\Omega_i} \in \Iso_{K^{\Omega_i}}(\mathfrak{x}^{\Omega_i},\mathfrak{y}^{\Omega_i})\right\}.\]
Here, $\mathfrak{x}^{\Omega_i}$ (respectively $\mathfrak{y}^{\Omega_i}$) denotes the restriction of the string $\mathfrak{x}$ (respectively~$\mathfrak{y}$) to the set $\Omega_i$.
Having computed the set $\Iso_{K^{\Omega_i}}(\mathfrak{x}^{\Omega_i},\mathfrak{y}^{\Omega_i})$ making one recursive call to String Isomorphism over domain size $n_i = |\Omega_i|$,
the set $\Iso_K^{\Omega_i}(\mathfrak{x},\mathfrak{y})$ can be computed in polynomial time by Theorem \ref{thm:permutation-group-algorithm-library}.
So overall the algorithm needs to make $s$ recursive calls to String Isomorphism over domain sizes $n_1,\dots,n_s$.

For the second type of recursion
let $H \leq G$ be a subgroup and let $T = \{g_1,\dots,g_t\}$ be a transversal for $H$. Then
\begin{equation}
 \Iso_{G}(\mathfrak{x},\mathfrak{y}) = \bigcup_{i \in [t]} \Iso_{Hg_i}(\mathfrak{x},\mathfrak{y}).
\end{equation}
In Luks's algorithm this type of recursion 
is applied when $G$ is a transitive group, $\mathfrak{B}$ is a minimal block system and $H = G_{(\mathfrak{B})}$.
Observe that $G^{\mathfrak{B}}$ is a primitive group and $t = |G^{\mathfrak{B}}|$.
Also note that $H$ is not transitive.
Indeed, each orbit of $H$ has size $n/b$ where $b = |\mathfrak{B}|$.
Hence, combining both types of recursion the computation of $\Iso_{G}(\mathfrak{x},\mathfrak{y})$ is reduced to $t\cdot b$ instances of String Isomorphism over domain size $n/b$.
We refer to this as the \emph{standard Luks reduction}.

Now suppose $G \in \ourgamma_d$. The crucial step to analyze Luks's algorithm is to determine the size of primitive groups occurring in the recursion.

\begin{theorem}[\cite{BCP82}]
 There exists a function $f$ such that every primitive $\ourgamma_d$-group $G \leq \Sym(\Omega)$ has order $|G| \leq n^{f(d)}$.
\end{theorem}

Indeed, the function $f$ can be chosen to be linear in $d$ (cf.\;\cite{LS99}).
As a result, Luks's algorithm runs in time $n^{\mathcal{O}(d)}$ for all groups $G \in \ourgamma_d$.

\subsection{Recursion}

For the purpose of later analyzing our recursion, we record some bounds.

\begin{lemma}
 \label{la:recursion-inductive-step}
 Let $k,n \in \mathbb{N}$
 and suppose $n_1,\dots,n_\ell \leq n/2$ such that $\sum_{i=1}^{\ell} n_i \leq 2^{k}n$.
 Then $\sum_{i=1}^{\ell} \left(\frac{n_i}{n}\right)^{k+1} \leq 1$.
\end{lemma}

\begin{proof}
 For $i \in [\ell]$ define $\alpha_i = \frac{n_i}{n}$. Observe that $\alpha_i \leq \frac{1}{2}$ and $\sum_{i=1}^{\ell} \alpha_i \leq 2^{k}$.
 Now suppose towards a contradiction that there are $\ell \in
 \mathbb{N}$ and nonnegative reals $\alpha_1,\dots,\alpha_\ell \in \mathbb{R}$  meeting these assumptions such that $\sum_{i=1}^{\ell} \alpha_i^{k+1} > 1$.
 Pick $\ell \in \mathbb{N}$, $\alpha_1,\dots,\alpha_\ell \in \mathbb{R}$ such that
 \begin{enumerate}
  \item[(i)] $\alpha_i \leq \frac{1}{2}$ for all $i \in [\ell]$ and $\sum_{i=1}^{\ell} \alpha_i \leq 2^{k}$,
  \item[(ii)] $\sum_{i=1}^{\ell} \alpha_i^{k+1} >1$,
  \item[(iii)] $\ell$ is minimal subject to Conditions (i) and (ii), 
  \item[(iv)] $|\{i \in [\ell] \mid \alpha_i = \frac{1}{2}\}|$ is
    maximal subject to Conditions (i) - (iii).
 \end{enumerate}
 Then $\alpha_i + \alpha_j > \frac{1}{2}$ for all $i,j \in [\ell]$.
 Let $A = \{i \in [\ell] \mid \alpha_i \neq \frac{1}{2}\}$ and suppose $|A| \geq 2$.
 Let $i,j \in A$ be distinct.
 Then
 \begin{equation*}
  \alpha_i^{k+1} + \alpha_j^{k+1} \leq \left(\frac{1}{2}\right)^{k+1} + \left(\alpha_i + \alpha_j - \frac{1}{2}\right)^{k+1}
 \end{equation*}
 which contradicts Condition (iv). Condition (iii) implies
 $\alpha_i>0$ for all $i$. Hence, $(\ell-1)\frac{1}{2}<\sum_{i=1}^\ell \alpha_i
 \leq 2^{k}$, which implies $\ell\le 2^{k+1}$.
 Therefore, $\sum_{i=1}^{\ell} \alpha_i^{k+1}  \leq \ell \left(\frac{1}{2}\right)^{k+1} \leq 1$, contradicting Condition (ii).
\end{proof}

\begin{lemma}
 \label{la:recursion-tree}
 Let $k\in\mathbb N$ and $t\colon \mathbb{N} \rightarrow \mathbb{N}$
 such that
 \begin{enumerate}
  \item $t(1) = 1$ and 
  \item for every $n \geq 2$ there are natural numbers $\ell \in \mathbb{N}$ and $n_1,\dots,n_\ell \leq n/2$ such that $t(n) \leq \sum_{i=1}^{\ell} t(n_i)$ and $\sum_{i=1}^{\ell} n_i \leq 2^{k}n$.
 \end{enumerate}
 Then $t(n) \leq n^{k+1}$ for all $n \in \mathbb{N}$.
\end{lemma}

\begin{proof}
 The statement is proved by induction on $n \in \mathbb{N}$.
 For $n > 1$ it holds that
 \begin{equation*}
  t(n) \leq \sum_{i=1}^{\ell} t(n_i) \leq \sum_{i=1}^{\ell} n_i^{k+1} = n^{k+1} \sum_{i=1}^{\ell} \left(\frac{n_i}{n}\right)^{k+1} \leq n^{k+1}
 \end{equation*}
 by Lemma \ref{la:recursion-inductive-step}.
\end{proof}

\begin{lemma}
  \label{la:recursion-bound}
 Let $k\in\mathbb N$ and $t\colon \mathbb{N} \rightarrow \mathbb{N}$ be a function such that $t(1) = 1$. Suppose that for every $n \in \mathbb{N}$ there are natural numbers~$n_1,\dots,n_\ell$ for which one of the following holds:
 \begin{enumerate}
  \item $t(n) \leq \sum_{i=1}^{\ell} t(n_i)$ where $\sum_{i=1}^{\ell} n_i \leq 2^{k}n$ and $n_i \leq n/2$ for all $i \in [\ell]$, or
  \item $t(n) \leq \sum_{i=1}^{\ell} t(n_i)$ where $\sum_{i=1}^{\ell} n_i \leq n$ and $\ell \geq 2$.
 \end{enumerate}
 Then $t(n) \leq n^{k+1}$.
\end{lemma}

\begin{proof}
 The statement is proved by induction on $n \in \mathbb{N}$.
 For the first option it holds that
 \begin{equation*}
  t(n) \leq \sum_{i=1}^{\ell} t(n_i) \stackrel{\text{I.H.}}{\leq} \sum_{i=1}^{\ell} n_i^{k+1} = n^{k+1} \sum_{i=1}^{\ell} \left(\frac{n_i}{n}\right)^{k+1} \leq n^{k+1}
 \end{equation*}
 by Lemma \ref{la:recursion-inductive-step}.
 For the second case we have
 \begin{equation*}
  t(n) \leq \sum_{i=1}^{\ell} t(n_i) \stackrel{\text{I.H.}}{\leq} \sum_{i=1}^{\ell} n_i^{k+1} \leq \left(\sum_{i=1}^{\ell} n_i\right)^{k+1} \leq n^{k+1}.
 \end{equation*}
\end{proof}

\begin{lemma}
 \label{la:approx-binom}
 Let $m,k \geq 1$ and suppose $k \leq \frac{m}{2}$. Then
 \begin{equation}
  \binom{m}{k}^{\log m} \geq m^{k}.
 \end{equation}
\end{lemma}

\begin{proof}
 It holds that \[\binom{m}{k}^{\log m} \geq \left(\frac{m}{k}\right)^{k \log m} \geq 2^{k \log m} = m^{k}.\]
\end{proof}

\section[The structure of primitive groups with restricted composition factors]{The structure of primitive groups in~$\ourgamma_d$}
\label{sec:characterization-primitive}

Recall that we denote by~$\ourgamma_d$ the class of groups whose composition factors are all isomorphic to subgroups of~$S_d$.
In this section we will prove several properties concerning the structure of primitive permutation groups in~$\ourgamma_d$ with a focus on their size in relation to their degree and~$d$.
More precisely, the goal of this section is to find a precise description of large primitive groups in~$\ourgamma_d$.
For the purpose of this work, a primitive permutation group $G \in \ourgamma_d$ is large if the cardinality of $G$ exceeds the term $n^{\mathcal{O}(\log d)}$.
We shall prove that large primitive permutation groups in~$\ourgamma_d$ are composed of Johnson groups in a well-defined manner meaning that Johnson groups form the only obstacles to efficient Luks reduction.
For the proof we perform a case-by-case analysis following the well-known O'Nan-Scott Theorem that classifies primitive permutation groups into five types.

\subsection{The O'Nan-Scott Theorem}

Let $G$ be a primitive permutation group acting on a set $\Omega$ of size $n$.
By the well known O'Nan-Scott Theorem (see, for example,~\cite{DM96}) the group $G$ has to be one of the following types.
The \emph{socle} of $G$, denoted by $\Soc(G)$, is the subgroup generated by all minimal normal subgroups of $G$.

\subparagraph{\textsf{I}. Affine Groups.} 
In this case there is a vector space~$V$ over a field of prime order~$p$ such that~$G$ is isomorphic to a group~$H$ that satisfies~$V^+\leq H\leq \AGL(V)$, where~$V^+$ is the additive group of the vector space~$V$.
The socle~$N$ of the group is a transitive abelian group (i.e,~$\mathbb{Z}_p^k$ for the prime~$p$ and an integer~$k$) and can be identified with~$V^+$. Furthermore, the stabilizer~$G_0$ of the~$0$-vector is an irreducible linear group (i.e., it does not have an invariant subspace).

\subparagraph{\textsf{II}. Almost Simple Groups.} In this case $\Soc(G) = T$ is a non-abelian simple group and $T \leq G \leq \Aut(T)$.

\subparagraph{\textsf{III}. Simple Diagonal Action.} In this case $\Soc(G) \cong T_1 \times \dots \times T_k$ where all $T_i$ are isomorphic to some non-abelian simple group $T$.
 Additionally, $n = |T|^{k-1}$, and the stabilizer of some point $\alpha \in \Omega$ is a diagonal subgroup $D \leq T_1 \times \dots \times T_k$.

\subparagraph{\textsf{IV}. Product Action.} In this case the set $\Omega$ can be identified with the $k$-tuples of some set~$M$. In particular $n = |M|^{k}$.
 Furthermore there is some primitive group $P \leq \Sym(M)$ of Type \textsf{II} or \textsf{III} and a transitive group $K \leq S_k$ such that $G \leq P \wr K$.
 The group $G$ acts in the natural product action of the wreath product. The socle of $G$ is $\Soc(G) = T^{k}$ where $T = \Soc(P)$.

\subparagraph{\textsf{V}. Twisted Wreath Product Action.} 
 In this case there is a transitive permutation group $P \leq S_k$ and a non-abelian simple group $T$ such that $G = B \rtimes P$ where $B$ is isomorphic to $T^{k}$.
 Furthermore $|\Omega| = |T|^{k}$ and $B$ acts regularly on $\Omega$.

\medskip

We analyze the structure of primitive~$\ourgamma_d$-groups according to the distinction into these five types. For each of them we will either be interested in a structural description or a bound on the size. To obtain such a bound we use the existence of small bases.

\begin{definition}
 Let $G \leq \Sym(\Omega)$ be a permutation group. A subset $B \subseteq \Omega$ is a \emph{base for $G$} if $G_{(B)} = \{1\}$.
 We define the minimum base size as $b(G) = \min\{|B| \mid B \subseteq \Omega\colon G_{(B)} = \{1\}\}$.
\end{definition}

The base size is related to the order of the group by the equation $2^{b(G)} \leq |G| \leq n^{b(G)}$.

\subsection{Affine type}

In the affine case we have a group $V^+\leq G\leq \AGL(V)$,
where~$V=\mathbb{F}_p^k$ and $V^+\cong\mathbb Z_p^k$ is the additive group of $V$,
such that the point stabilizer~$G_0$ is irreducible.
We say that a group~$G_0\leq \GL(k,p)$ \emph{acts primitively as a linear group}
if it does not preserve any direct sum decomposition $V = V_1 \oplus \dots \oplus V_\ell$, $\ell \geq 2$, of the underlying vector space~$V=\mathbb{F}_p^k$.
For primitive affine groups in~$\ourgamma_d$ we will draw conclusions using a characterization of~\cite{LS02, LS14}.
The characterization involves quasi-simple classical
groups~$\SL_{r}(q')$, $\SU_{r}(q')$, $\Sp_{r}(q')$ or~$\Omega_{r}(q')$.
Recall that a group is quasi-simple if it is equal to its own commutator subgroup (i.e., perfect) and it is simple modulo its center.
With finitely many exception, the mentioned groups are indeed quasi-simple~\cite[Proposition 1.10.3]{BHRC13}.

\begin{theorem}[Consequence of \cite{LS02}, \cite{LS14}]\label{thm:prim:linear:grp:struct}
 There are absolute constants $c_1,c_2,c_3 \in \mathbb{N}$ such that the following holds.
 Let $G \leq \Sym(\Omega)$ be a primitive group of Type \textsf{I}
 such that $G_0$
 acts primitively as a linear group.
 Then
 \begin{enumerate}
  \item $b(G) \leq c_1$, or\label{item:bounded:case}
  \item $G$ contains a quasi-simple classical group of rank $k$, more precisely $\SL_k(q')$,~$\SU_k(q')$,~$\Sp_k(q')$ or~$\Omega_k(q')$, and $b(G) \leq c_2 k + c_3$, or\label{item:lie:group:case}
  \item $G$ contains an alternating group $A_k$ and $b(G) \leq c_2 \log k + c_3$.\label{item:alternating:case}
 \end{enumerate}
\end{theorem}

\begin{proof} Let $G \leq \Sym(\Omega)$ be a primitive group of Type \textsf{I} such that $G_0$
 acts primitively as a linear group.

Theorem~1 in~\cite{LS14} (which applies to~$G$, see the comment after that theorem) states that either~$b(G)\leq C$ or that~$F^*(H^0)\leq G$ contains~$\prod_{i=1}^s \Alt(m_i)\cdot \prod_{i=1}^t \Cl_{d_i}(q_i)^{(\infty)}$ for some integers~$m_1,\ldots,m_s,q_1,\ldots q_t$.  Here~$\Cl_{(d_i)}(q_i)$ is the normalizer of a quasi-simple classical group, namely $\SL_{d_i}(q')$, $\SU_{d_i}(q')$, $\Sp_{d_i}(q')$ or~$\Omega_{d_i}(q')$.
Furthermore~$\Cl_{(d_i)}(q_i)^{(\infty)}$ is the last group of the derived series of~$\Cl_{(d_i)}(q_i)$.
Only finitely many of the groups $\SL_{d_i}(q')$,~$\SU_{d_i}(q')$,~$\Sp_{d_i}(q')$ or~$\Omega_{d_i}(q')$ are not perfect.
Thus, by referring to Case 1 for the finitely many exceptions, we can assume that $\Cl_{(d_i)}(q_i)^{(\infty)}$ contains $\SL_{d_i}(q')$,~$\SU_{d_i}(q')$,~$\Sp_{d_i}(q')$ or~$\Omega_{d_i}(q')$.

Proposition~2 in~\cite{LS14} states further that either~$b^*(H^0)\leq 9d_i+22$ for all~$i$ or~$b^*(H^0) \leq 3\log_p m_i+22$ for all~$m_i$. 
Also~$b(G)\leq b(H^0)+1\leq b^*(H^0)+2$.
\end{proof}

\begin{lemma}[\cite{Cooperstein78},~\cite{KleLie90}]\label{lem:rank:vs:size:d}
 Let $G \in \ourgamma_d$ be a simple group of Lie type of rank~$k$ or one of the quasi-simple classical groups $\SL_k(q')$,~$\SU_k(q')$,~$\Sp_k(q')$ or~$\Omega_k(q')$.
 Then $k = \mathcal{O}(\log d)$.
\end{lemma}

\begin{proof}
Note that by our definition of $\ourgamma_d$, a simple group is in~$\ourgamma_d$ if and only if it is a subgroup of~$S_d$. To prove the lemma it thus suffices to show that the smallest~$d(k)$ for which~$S_{d(k)}$ contains a simple group of Lie type of rank~$k$ is exponential in~$k$.

Cooperstein~\cite{Cooperstein78} lists the minimum degree of a permutation representation of the mentioned quasi-simple classical groups. They are all exponential in the rank~$k$. In~\cite{KleLie90}    the minimum degree of a permutation representation is listed for all simple groups of Lie type. Likewise 
they are exponential in~$k$.
\end{proof}

\begin{theorem}
 \label{cor:primitive-size-type-1}
 Let $G \in \ourgamma_d$ be a primitive permutation group of degree $n$ of Type \textsf{I}.
 Then $b(G) = \mathcal{O}(\log d)$ and therefore $|G| = n^{\mathcal{O}(\log d)}$. 
\end{theorem}

\begin{proof}
Let~$G\in \ourgamma_d$ be a primitive permutation group of Type \textsf{I}. This implies that the point stabilizer~$G_0\leq \GL(k,p)$, where~$p^k=n$, is an irreducible linear group.
It suffices to show that~$b(G_0)\in  \mathcal{O}(\log d)$ since~$b(G)=b(G_0)+1$.

If~$G_0$ is primitive as a linear group this follows by assembling Theorem~\ref{thm:prim:linear:grp:struct} and Lemma~\ref{lem:rank:vs:size:d}: indeed, if~$G_0$ is in Case~\ref{item:bounded:case} of Theorem~\ref{thm:prim:linear:grp:struct} the claim is obvious. For Case~\ref{item:lie:group:case} we see that~$b(G)\leq c_2k+c_3\in \mathcal{O}(\log d)$ by the lemma and for the last case we know that~$b(G)\leq c_2 \log k+c_3\in \mathcal{O}(\log d)$.

Now suppose that~$G_0$ is irreducible but imprimitive. Then~$G_0$ can be written as~$P\wr H$ for some primitive linear group~$P\leq \GL(k/\ell,p)$ and transitive group~$H$ that permutes~$\ell$ subspaces~$V_1,\ldots,V_{\ell}$ of~$V^k$.
By~\cite[Lemma 4.2 (a)]{GSS98} there is a set~$B_1$ of~$\mathcal{O}(\log d)$ points in~$V^k$ such that for the point-wise stabilizer we have~$(G_0)_{(B_1)} \leq P^{\ell}$.

(We now follow the techniques from Section~6 of~\cite{GSS98}.)
Since~$P$ is a primitive linear group, by the first part of the proof there is a base~$\{x_1\ldots,x_t\}$ of~$P$ of size $t =\mathcal{O}(\log d)$ .
Let~$b_i$ be the point~$(x_i,x_i,\ldots,x_i)$ in~$V_1\times V_2\times \cdots \times V_{\ell}$.
Then~$B_2= \{b_1,\ldots,b_t\}$ is a base of~$P^{\ell}$.
If follows that~$B_1\cup B_2$ is a base for~$G_0$ of size at most~$\mathcal{O}(\log d)$.
\end{proof}

\subsection{Non-affine type}

For a group $G$ we denote by $\Out(G)$ the outer automorphism group of $G$.
It is well-known that $|\Out(S_m)|=1$ and $|\Out(A_m)|=2$ for all $m>6$.

\begin{lemma}
 Let $G$ be a non-abelian simple group. Then $|\Out(G)| \leq \mathcal{O}(\log |G|)$.
\end{lemma}

\begin{proof}
For finite simple groups of Lie type, this follows by inspecting the Tables~5 and~6 in the Atlas of Finite Groups~\cite{Atlas}.
In Table 5 the size of the outer automorphism group is given as the
product~$d\cdot  f \cdot g$ which, according to Table~6, is
logarithmic in the size of the group for each simple group of Lie
type. For alternating simple groups the statement is obvious. 
The values for the sporadic groups disappear in the~$\mathcal{O}$-notation.
\end{proof}

Recall that for $t \leq \frac{m}{2}$ we denote by $A_m^{(t)}$ be the action of the alternating group $A_m$ on the set of $t$-element subsets of $[m]$.

\begin{theorem}[Liebeck \cite{liebeck84}]
 \label{thm:large-primitive-almost-simple-classification}
 Let $G \leq \Sym(\Omega)$ be a primitive group and suppose $N = \Soc(G)$ is simple. Then one of the following holds:

 \begin{enumerate}
  \item\label{item:liebeck-almost-simple-1} $N$ is permutationally equivalent to $A_m^{(t)}$ for some $m \in \mathbb{N}$ and $t \leq \frac{m}{2}$,
  \item\label{item:liebeck-almost-simple-2} $N$ is permutationally
    equivalent to $A_m$ acting on the set of partitions of $[m]$ into
    subsets of size $b$ (for some $b \leq m$),
  \item\label{item:liebeck-almost-simple-3} $N$ is a classical simple
    group acting on an orbit of subspaces of the natural module or
    pairs of subspaces of complementary dimension, or
  \item $|G| \leq n^{9}$.
 \end{enumerate}
\end{theorem}

We will not exploit the structure of the action of~$N$ in
Case~\ref{item:liebeck-almost-simple-3} of the Theorem, and rather
only use that~$N$ is a simple group of Lie type.

\begin{lemma}
 \label{la:primitive-size-type-2}
 Let $G \leq \Sym(\Omega)$ be a primitive $\ourgamma_d$-group of Type \textsf{II}.
 Let $N = \Soc(G)$. Then one of the following holds:
 \begin{enumerate}
  \item\label{item:primitive-size-type-2-item-1} $N$ is permutationally equivalent to $A_m^{(t)}$ for some $m \leq d$ and $t \leq \frac{m}{2}$ and $|G:N| \leq 2$, or
  \item $|G| = n^{\mathcal{O}(\log d)}$.
 \end{enumerate}
\end{lemma}

\begin{proof}
 The proof is based on Theorem \ref{thm:large-primitive-almost-simple-classification}.
 First suppose $N$ is permutationally equivalent to~$A_m^{(t)}$ for some $m \in \mathbb{N}$ and $t \leq \frac{m}{2}$.
 Note that $m \leq d$ since $N \in \ourgamma_d$ by Lemma \ref{la:gamma-d-closure}.
 Furthermore $|G:N| \leq |\Out(N)| \leq 2$  since~$N$ is an alternating group (in case $m \leq 6$ the second option is satisfied).
 
 Next consider the case that $N$ is permutationally equivalent to $A_m$ acting on partitions of $[m]$ into subsets of size $b$.
 Again, $m \leq d$ and $|G:N| \in \mathcal{O}(1)$.
 In this case $n = \frac{m!}{(b!)^{a}a!}$ where $a \cdot b = m$.
 Using Stirling's approximation it can be calculated that $n = 2^{\Omega(m)}$.
 Hence, $|N| \leq m^{m} = n^{\mathcal{O}(\log m)} = n^{\mathcal{O}(\log d)}$ and consequently $|G| = n^{\mathcal{O}(\log d)}$.
 
 It remains to analyze the third case. 
 It suffices to show that~$|N| \in  n^{\mathcal{O}(\log d)}$ since then~$|G|\leq |N| |G:N|\leq |N| |\Out(N)|\leq |N| \mathcal{O}(\log (|N|)) \in  n^{\mathcal{O}(\log d)}$.
 For this let~$\varphi\colon N\rightarrow S_{d'}$ be a permutation representation of~$N$ with~$d'$ as small as possible.
 Then~$d'\leq d$ and~$d'\leq n$. Moreover, being minimal, the action is faithful and primitive since~$N$ is simple.
 Not being an alternating group, the group~$N$ is not a Cameron group\footnote{Cameron groups form a collection of primitive permutation groups
 that exactly characterize primitive permutation groups of size greater than $n^{1 + \log n}$ (for $n$ sufficiently large) \cite{Cameron81,Mar02}. We shall not formally define Cameron groups in this work and rather only remark that the only simple Cameron groups are alternating groups (not necessarily in their standard action).}
 and we conclude~$|N|\leq (d')^{1+ \log {(d')}} \in n^{\mathcal{O}(\log d)}$~\cite{Cameron81,Mar02}.
\end{proof}

\begin{lemma}[\cite{GSS98}]
 \label{la:primitive-size-type-3}
 Let $G \in \ourgamma_d$ be a primitive group of Type \textsf{III}.
 Then $b(G)\leq 2\ell+1$ and thus~$|G| \leq n^{2\ell+1}\in n^{\mathcal{O}(\log d)}$ where $\ell := \max\{5,\lceil \log d\rceil\}$.
\end{lemma}

\begin{proof}[Remark on the proof]
 While not explicitly stated in~\cite{GSS98}, the Lemma is implicit from \cite[Lemma 4.2 (c)]{GSS98} and the comment in Section 6 on Type~III~(a) in~\cite{GSS98}. (We advise that the types in that paper do not agree with ours.)
\end{proof}

\begin{lemma}
 \label{la:primitive-size-type-4}
 Let $G \leq \Sym(\Omega)$ be a primitive $\ourgamma_d$-group of Type \textsf{IV}.
 Let $N = \Soc(G)$. Then one of the following holds:
 \begin{enumerate}
  \item\label{item:primitive-size-type-4-item-1} 
  $G \leq P \wr K$ is a wreath product in the product action of a transitive group $K\leq S_k$ in~$\ourgamma_d$ and a group~$P$ of Type~\textsf{II} with a socle~$T$ permutationally equivalent to $A_m^{(t)}$ for some $m \leq d$ and $t \leq \frac{m}{2}$, and $N$ is 
  isomorphic to~$T^k$ with $|G:N| \leq n^{1 + \log d}$, 
  or
  \item $|G| = n^{\mathcal{O}(\log d)}$.
 \end{enumerate}
\end{lemma}

\begin{proof}
 If~$G\in \ourgamma_d$ is of Type~\textsf{IV} then $G \leq P \wr K$ for some primitive group $P \leq \Sym(M)$ and a transitive group $K \leq S_k$.
 Observe that both $P \in \ourgamma_d$ and $K \in \ourgamma_d$.
 Let $H = P^{k}$.
 Then $|G:H| = |K| \leq d^{k-1} \leq 2^{k \cdot \log d} \leq n^{\log d}$ by Lemma \ref{la:size-gamma-d}.
 Moreover $|G| \leq n^{\log d} \cdot |P|^{k}$.
 
 In case $|P| = |M|^{\mathcal{O}(\log d)}$ we thereby conclude $|G| \leq n^{\log d} \cdot (|M|^{\mathcal{O}(\log d)})^{k} = n^{\mathcal{O}(\log d)}$.
 Thus, by Lemmas \ref{la:primitive-size-type-3} and \ref{la:primitive-size-type-2}, it suffices to consider the case where $P$ is a primitive group of Type \textsf{II} which satisfies Part \ref{item:primitive-size-type-2-item-1} of Lemma \ref{la:primitive-size-type-2}. That is, the socle $T = \Soc(P)$ of~$P$ 
 is permutationally equivalent to $A_m^{(t)}$ for some $m \leq d$ and $t \leq \frac{m}{2}$ and $|P:T| \leq 2$.
 Note that $N = T^{k} \leq P^k= H$.
 Thus, $|G:N| = |G:H| \cdot |P:T|^{k} \leq n^{\log d} \cdot 2^{k} \leq n^{1 + \log d}$.
\end{proof}

\begin{lemma}
 \label{la:primitive-size-type-5}
 Let $G \in \ourgamma_d$ be a primitive group of Type \textsf{V}.
 Then $|G| \leq n^{1+\log d}$.
\end{lemma}

\begin{proof}
 For a primitive group~$G$ of Type \textsf{V}, a primitive twisted wreath product, 
 there is a transitive group $P \leq S_k$ and a non-abelian simple group $T$ such that $G \cong T^k \rtimes P$.
 Moreover, $n=|\Omega| = |T|^{k}$ and thus~$k\leq \log(n)$.
 Note that $P \in \ourgamma_d$ since~$\ourgamma_d$ is closed under subgroups and thus~$|P|\leq d^{k-1}$ by Lemma \ref{la:size-gamma-d}. We conclude
 that $|G| = |T^k| \cdot |P| = n \cdot |P| \leq n \cdot d^{k-1} \leq n \cdot d^{\log n} = n^{1 + \log d}$.
\end{proof}

\subsection[Structure theorem for primitive groups with restricted composition factors]{Structure theorem for primitive groups in~$\ourgamma_d$}

Having analyzed the structure of large $\ourgamma_d$-groups for all five types of primitive groups we now combine those statements into a structure theorem.
For this, we need two auxiliary lemmata.

\begin{lemma}
 \label{la:primitive-stabilizer}
 Let $G \leq \Sym(\Omega)$ be a transitive group and $\alpha \in \Omega$.
 Then \[B_\alpha = \{\beta \in \Omega \mid \beta^{G_\alpha} = \{\beta\}\}\] forms a block of $G$.
\end{lemma}

\begin{proof}
 Let $R = \{(\alpha,\beta) \in \Omega^{2} \mid \beta^{G_\alpha} = \{\beta\}\}$.
 Clearly the relation $R$ is reflexive and transitive.
 Suppose that $(\alpha,\beta) \in R$.
 Then $G_\alpha \leq G_\beta$.
 Moreover, $|G_\alpha| = |G|/|\alpha^{G}| = |G|/|\beta^{G}| = |G_\beta|$ since $G$ is transitive.
 It follows that $G_\alpha = G_\beta$ and thus, $(\beta,\alpha) \in R$.
 So $R$ is also symmetric and hence, $R$ is an equivalence relation.
 
 Now let $(\alpha,\beta) \in R$ and $g \in G$.
 Then $(\alpha^{g},\beta^{g}) \in R$ because $G_{\alpha^{g}} = g^{-1}G_\alpha g$.
 Thus, $R$ is invariant under $G$ and the partition into equivalence classes forms a block system for $G$.
\end{proof}

\begin{lemma}
 \label{la:direct-product-all-block-systems}
 Let $P \leq \Sym(\Omega)$ be a non-regular primitive group and $k \geq 2$.
 Let $\mathfrak{B}$ be a block system of $P^{k}$ with its natural
 action on $\Omega^{k}$.
 Then there is some $I \subseteq [k]$ such that
 \[\mathfrak{B} = \{\{(\alpha_1,\dots,\alpha_k) \in \Omega^{k} \mid \forall i \in I\colon \alpha_i = \beta_i\} \mid (\beta_i)_{i \in I} \in \Omega^{|I|}\}.\]
\end{lemma}

\begin{proof}
 Let $B \in \mathfrak{B}$ be a block and let $I = \{i \in [k] \mid |\pi_i(B)| = 1\}$ where $\pi_i(B) = \{\alpha_i \mid (\alpha_1,\dots,\alpha_k) \in B\}$.
 For every $i \in I$ suppose $\pi_i(B) = \{\beta_i\}$.
 It suffices to show that $B = \{(\alpha_1,\dots,\alpha_k) \in \Omega^{k} \mid \forall i \in I\colon \alpha_i = \beta_i\}$.
 Let $j \in [k] \setminus I$ and let $(\alpha_1,\dots,\alpha_k) ,(\alpha_1',\dots,\alpha_k') \in B$ such that $\alpha_j \neq \alpha_j'$.
 Since $G$ is non-regular and primitive there is some $g \in G_{\alpha_j}$ such that $(a_j')^{g} \neq \alpha_j'$ (see Lemma \ref{la:primitive-stabilizer}).
 Note that $(\alpha_1',\dots,\alpha_{j-1}',
 (a_j')^{g},\alpha_{j+1}',\alpha_k') \in B$.
 Let $\Delta = \{\alpha \in \Omega \mid (\alpha_1',\dots,\alpha_{j-1}', \alpha,\alpha_{j+1}',\alpha_k') \in B \}$.
 Since $\Delta$ forms a block of $P$ and $|\Delta| \geq 2$ we get that $\Delta = \Omega$.
 This implies that $B = \{(\alpha_1,\dots,\alpha_k) \in \Omega^{k} \mid \forall i \in I\colon \alpha_i = \beta_i\}$.
\end{proof}

Let $G \leq \Sym(\Omega)$ and let $\mathfrak{B},\mathfrak{B}'$ be two $G$-invariant partitions such that $\mathfrak{B} \prec \mathfrak{B}'$.
Consistent with our previous notation we denote by $G_B^{\mathfrak{B}[B]}$ the natural induced action of $G_B$ on the set $\mathfrak{B}[B]$ for all $B \in \mathfrak{B}'$.

\begin{theorem}
 \label{thm:first-main-theorem}
 Let $G \leq \Sym(\Omega)$ be a primitive $\ourgamma_d$-group.
 Then one of the following holds:
 \begin{enumerate}
  \item\label{item:first-main-theorem-1} $|G| \leq n^{c_1 \log d + c_2}$ for some absolute constants $c_1,c_2$, or
  \item\label{item:first-main-theorem-2} for the normal subgroup $N = \Soc(G) \leq G$ there is a sequence of partitions $\{\Omega\} = \mathfrak{B}_1 \succ \dots \succ \mathfrak{B}_k = \{\{\alpha\} \mid \alpha \in \Omega\}$ such that the following holds:
  \begin{enumerate}
   \item $|G:N| \leq n^{1 + \log d}$,
   \item $\mathfrak{B}_i$ is $N$-invariant for every $i \in [k]$, and
   \item there are $m \leq d$ and $t \leq \frac{m}{2}$ with $m > 4 \log s$ where $s = \binom{m}{t}$ such that for all $i \in [k-1]$ and $B \in \mathfrak{B}_i$ the group $N_B^{\mathfrak{B}_{i+1}[B]}$
         is permutationally equivalent to $A_m^{(t)}$.
  \end{enumerate}
 \end{enumerate}
 Moreover, there is a polynomial-time algorithm that determines one of the options that is satisfied and in case of the second option computes~$N$ and the partitions~$\mathfrak{B}_1, \dots ,\mathfrak{B}_k$.
\end{theorem}

\begin{proof}
 First suppose $G$ is a primitive group of Type \textsf{I}, \textsf{III} or \textsf{V}. 
 Then 1 holds; the claimed bound on the group size follows from
 Corollary \ref{cor:primitive-size-type-1}, Lemma \ref{la:primitive-size-type-3} and \ref{la:primitive-size-type-5}, respectively.
 So it remains to consider primitive groups of Type \textsf{II} and \textsf{IV}.
 
 Let $N = \Soc(G)$ be the socle of $G$. Suppose $G$ is a primitive
 group of Type \textsf{II}.  Then by
 Lemma~\ref{la:primitive-size-type-2} we conclude
 $|G| = n^{\mathcal{O}(\log d)}$ or $N$ is permutationally equivalent
 to $A_m^{(t)}$ for some $m \leq d$ and $t \leq \frac{m}{2}$ and
 $|G:N| \leq 2$.  In case $m \leq 4 \log \binom{m}{t} = 4 \log n$, it
 holds that $|N|\leq m^m \leq n^{4 \log m} \leq n^{4 \log d}$ and
 thus, $|G| = n^{\mathcal{O}(\log d)}$.  In
 case~$m > 4 \log \binom{m}{t}$ we set~$\mathfrak{B}_1 = \{\Omega\}$
 and~$\mathfrak{B}_2 = \{\{\alpha\} \mid \alpha \in \Omega\}$.
 
 Next assume~$G$ is of Type~\textsf{IV}. By
 Lemma~\ref{la:primitive-size-type-4}, $G \leq P \wr K$ is a wreath
 product in the product action for a transitive group $K\leq S_k$
 in~$\ourgamma_d$ and a group~$P\leq \Sym(M)$
 of Type~\textsf{II} with a socle~$T$ permutationally equivalent to
 $A_m^{(t)}$ for some $m \leq d$ and $t \leq \frac{m}{2}$, and $N$ is
 isomorphic to~$T^k$ with $|G:N| \leq n^{1 + \log d}$.  Moreover, in
 case $m \leq 4 \log |M|$ we have $|T|\le m! \leq |M|^{4 \log m}$
 and hence, $|N| \leq n^{4 \log d}$.  This implies that
 $|G| = n^{\mathcal{O}(\log d)}$, so we can assume
 $m > 4 \log |M|= 4 \log \binom{m}{t}$.

 Observe that, since the wreath product is in the product action, an element $h=(p_1,\dots,p_k) \in P^k$ acts on an element $(m_1,\dots,m_k) \in M^{k} = \Omega$ via $(m_1,\dots,m_k)^{h} = (m_1^{p_1},\dots,m_k^{p_k})$.
 For $i \in [k+1]$ define
 \[\mathfrak{B}_i = \{\{(m_1,\dots,m_k) \in M^{k} \mid \forall j < i\colon m_j = m_j^{*} \} \mid m_1^{*},\dots,m_{i-1}^{*} \in M\}.\]
 Clearly, $\mathfrak{B}_i$ is an $N$-invariant partition for all $i \in [k+1]$.
 Observe that $N_B^{\mathfrak{B}_{i+1}[B]}$ is permutationally equivalent to $T$ for all $i \in [k]$ which itself is permutationally equivalent to $A_m^{(t)}$.
 
 We describe a polynomial time algorithm as required by the theorem. Note first that~$|G|$ can be computed in polynomial time, so Option~1 can be detected. 
 Also note that the socle of a group is a normal subgroup and can be computed in polynomial time (see\;\cite{KL90}).

 The algorithm now sets~$\mathfrak{B}_1 =\{\Omega\}$.
 To compute~$\mathfrak{B}_{i+1}$ from~$\mathfrak{B}_{i}$ we choose an arbitrary block~$B\in \mathfrak{B}_{i}$ and compute a maximal block~$B'$ within~$B$, that is a block that is inclusion wise maximal with the property that~$B'\subsetneq B$.
 We set~$\mathfrak{B}_{i+1}= (B')^N$.

 Note that, up to permuting the coordinates, by Lemma \ref{la:direct-product-all-block-systems} the block systems described above are the only block systems of $N$.
 Hence every sequence of block systems~$\{\Omega\} = \mathfrak{B}_1 \succ \dots \succ \mathfrak{B}_k = \{\{\alpha\} \mid \alpha \in \Omega\}$ that cannot be extended has the desired properties.
 
 Finally note that the algorithm is also correct for groups of Type~\textsf{II}, since then~$N$ is primitive and we get the sequence $\mathfrak{B}_1 = \{\Omega\}$ and~$\mathfrak{B}_2 = \{\{\alpha\} \mid \alpha \in \Omega\}$.
\end{proof}

\begin{remark}
 Let $\Gamma_d$ denote the family of groups $G$ with the property that $G$ has no alternating composition factors of degree greater than $d$
 and no classical composition factors of rank greater than $d$.
 (There is no restriction on the cyclic, exceptional, and sporadic composition factors of $G$.)
 While the class $\ourgamma_d$ considered in this paper follows the original definition of Luks \cite{luks82},
 most of the recent literature is concerned with the more general class of groups $\Gamma_d$ \cite{BCP82,GSS98}.
 The reason is that many results that can be proved for the class $\ourgamma_d$ indeed carry over to the more general class of groups $\Gamma_d$.
 We want to stress the fact that this is not the case for the results presented in this section.
 Indeed, consider the affine general linear group $G = \AGL(d,p)$ of dimension $d$ (with its natural action on the corresponding vector space).
 Then $G$ is a primitive group of affine type and $|G| = n^{\Omega(d)}$ where $n = p^{d}$ is the size of the vector space.
 For this group Theorem \ref{thm:first-main-theorem} does not hold.
 The group $G$ is contained in the class $\Gamma_d$, but it is not contained in $\ourgamma_d$.
\end{remark}

\section[Almost d-ary block system sequences]{Almost~$d$-ary block system sequences}\label{sec:almost:d:ary}

In the previous section we essentially proved that the only obstacles to efficient Luks reduction are Johnson groups which is very similar to Babai's quasipolynomial time algorithm.
Hence, the natural approach to tackle the obstacle cases seems to be an adaption of the techniques introduced by Babai \cite{Babai15-full,Babai16}.
However, there is an intrinsic problem.
The group-theoretic methods forming the basis for Babai's algorithm rest on a group-theoretic statement, the Unaffected Stabilizers Theorem, for which the natural adaption to our setting does not hold (cf.\ \cite[Remark 8.2.5]{Babai15-full}).
To remedy this problem we introduce a preprocessing step that reduces the String Isomorphism Problem for $\ourgamma_d$-groups to a more restricted version of this problem.
In this restricted version, the group is equipped with a sequence of block systems satisfying a particular property defined as follows.  
(Recall that a permutation group $G \leq \Sym(\Omega)$ is semi-regular
if $G_\alpha = \{1\}$ for every $\alpha \in \Omega$.
Also remember that, for $G$-invariant partitions $\mathfrak{B} \prec \mathfrak{B}'$ and $B \in \mathfrak{B}'$,
we denote by $G_B^{\mathfrak{B}[B]}$ the natural induced action of $G_B$ on the set $\mathfrak{B}[B]$.)

\begin{definition}
 Let $G \leq \Sym(\Omega)$ be a permutation group. A $G$-invariant sequence of partitions
 $\{\Omega\} = \mathfrak{B}_0 \succ \dots \succ \mathfrak{B}_k = \{\{\alpha\} \mid \alpha \in \Omega\}$ 
 is called \emph{almost $d$-ary} if for every $i \in [k]$ and $B \in \mathfrak{B}_{i-1}$ it holds that 
 \begin{enumerate}
  \item $G_B^{\mathfrak{B}_{i}[B]}$ is semi-regular, or
  \item $|\mathfrak{B}_{i}[B]| \leq d$.
 \end{enumerate}  
 The sequence is called \emph{$d$-ary} if the second property is satisfied for every $i \in [k]$ and $B \in \mathfrak{B}_{i-1}$.
\end{definition}

A simple, but crucial observation is that such sequences are inherited by subgroups and restrictions to invariant subsets.

\begin{observation}
 \label{obs:sequence-of-partitions}
 Let $G \leq \Sym(\Omega)$ be a group, and let $\{\Omega\} =
 \mathfrak{B}_0 \succ \dots \succ \mathfrak{B}_m = \{\{\alpha\} \mid
 \alpha \in \Omega\}$ be an (almost) $d$-ary sequence of $G$-invariant partitions.
 Moreover, let $H \leq G$.
 Then $\mathfrak{B}_0 \succ \dots \succ \mathfrak{B}_m$ also forms an (almost) $d$-ary sequence of $H$-invariant partitions.
 Additionally, for an $H$-invariant subset $\Delta \subseteq \Omega$ it holds that $\mathfrak{B}_0[\Delta] \succeq \dots \succeq \mathfrak{B}_m[\Delta]$ forms an (almost) $d$-ary sequence of $H^{\Delta}$-invariant partitions.
\end{observation}

For groups equipped with an almost $d$-ary sequence of partitions it is possible to give a natural variant of the Unaffected Stabilizers Theorem which, in turn, allows for an adaption of Babai's algorithmic techniques to give an efficient algorithm deciding String Isomorphism for this type of groups.

The goal of this section is to describe a reduction that, given an instance of String Isomorphism for $\ourgamma_d$-groups, computes a new equivalent instance, in which the permutation group is equipped with an almost \emph{$d$-ary} $G$-invariant sequence of partitions.
This reduction runs in time $n^{\polylog{d}}$ and builds on the classification of large primitive groups obtained in the previous section.
We shall then see in subsequent sections that the String Isomorphism Problem for groups equipped with such a sequence can be solved in time $n^{\polylog{d}}$.

\subsection{The high-level idea}

The central idea for the reduction is to change the action of the permutation group $G$.
More precisely, we shall construct a new permutation domain $\Omega^{*}$ and consider an action of the group $G$ on the set $\Omega^{*}$.
Indeed, the set $\Omega^{*}$ will be larger than the original permutation domain $\Omega$.
Note that this is acceptable for our purposes as long as $|\Omega^{*}| \leq |\Omega|^{\polylog{d}}$.

Let us first illustrate this on a high level for the special case that $G$ is a primitive group.
Using the characterization of primitive $\ourgamma_d$-groups given in the previous section we have to distinguish two cases.
First suppose that $|G| \leq n^{c_1 \log d + c_2}$ for some appropriate absolute constants $c_1,c_2$.
Now define $\Omega^{*} = G \times \Omega$.
Then $g \in G$ acts on $\Omega^{*}$ via
\[(h,\alpha)^{g} = (hg,\alpha^{g}).\]
Let $G^{*} \leq \Sym(\Omega^{*})$ be the permutation group obtained from the action of $G$ on the set $\Omega^{*}$.
It is easy to check that $G^{*}$ is semi-regular.
Also note that $|\Omega^{*}| \leq n^{\mathcal{O}(\log d)}$.
Of course we also need to transform the strings.
For a string $\mathfrak{x}\colon \Omega \rightarrow \Sigma$ define
$\mathfrak{x}^{*}\colon \Omega^{*} \rightarrow \Sigma\colon (h,\alpha)
\mapsto \mathfrak{x}(\alpha)$. 
Note that no information is lost in this transformation.
Indeed, it can be verified that two strings $\mathfrak{x},\mathfrak{y}$ are $G$-isomorphic if and only if $\mathfrak{x}^{*}$ is $G^{*}$-isomorphic to $\mathfrak{y}^{*}$.
So this gives us the desired reduction.

Next, let us consider the more interesting case that $G$ satisfies Property \ref{item:first-main-theorem-2} of Theorem \ref{thm:first-main-theorem}.
Let $N = \Soc(G)$.
Then, in a first step, we consider the set $\Omega^{*} = G/N \times \Omega$.
An element $g \in G$ acts on $\Omega^{*}$ via
\[(Nh,\alpha)^{g} = (Nhg,\alpha^{g}).\]
Let $G^{*} = G^{\Omega^{*}} \leq \Sym(\Omega^{*})$ denote the permutation group corresponding to the action of $G$ on $\Omega^{*}$.
Now the crucial observation is that $\mathfrak{B} = \{\{(Nh,\alpha) \mid \alpha \in \Omega\} \mid h \in G\}$ is a $G^{*}$-invariant partition.
For every $B \in \mathfrak{B}$, it holds that $\left(G^{*}\right)_B^{B}$ is permutationally equivalent to $N$, and the group $(G^{*})^{\mathfrak{B}}$ is regular.
Note that again $|\Omega^{*}| \leq n^{\mathcal{O}(\log d)}$.
Also, the strings can be transformed in the same way as before.
Hence, it remains to consider only the group $N$.

Finally, for an intuition on how the group $N$ is transformed suppose for simplicity that $N = A_m^{(t)}$.
The group $A_m$ has another action closely related to the action
$A_m^{(t)}$ on the $t$-element subsets of $[m]$, 
 namely the action on the set $[m]^{\langle t\rangle}$ of all $t$-tuples with distinct entries.
A crucial difference between these actions is that the action on the tuples is not primitive.
Indeed, fixing more and more coordinates, we get the following sequence of partitions.
For $i = 0,\dots,t$ let
\[\mathfrak{B}_i^{*} = \{\{(a_1,\dots,a_t) \in [m]^{\langle t\rangle}
  \mid \forall j \leq i\colon a_j = b_j\} \mid (b_1,\dots,b_i) \in [m]^{\langle i\rangle}\}.\]
Let $N^{*}$ be the action of $N$ on the set of ordered $t$-tuples with distinct entries.
For every $i \in [t]$ the partition $\mathfrak{B}_i^{*}$ is $N^{*}$-invariant and for every $B \in \mathfrak{B}_{i-1}^{*}$ it holds that $|\mathfrak{B}_{i}^{*}[B]| \leq m \leq d$.
Moreover, with every element $\bar a \in [m]^{\langle t\rangle}$ we can associate the underlying unordered set of elements.
This way, we can also transform the strings in a way similar to before.
Also note that the set $[m]^{\langle t\rangle}$ is only slightly larger than $\binom{m}{t}$ (cf.\ Lemma \ref{la:approx-binom}).

In the following we describe this reduction for general groups.
Analogous to this high level description we proceed in two steps.

\subsection{First Step}

\begin{theorem}
 \label{thm:reduction-one}
 Let $G \leq \Sym(\Omega)$ be a transitive $\ourgamma_d$-group and let $\mathfrak{x},\mathfrak{y}\colon\Omega\rightarrow\Sigma$ be two strings.
 Then there is a set $\Omega^{*}$, a $\ourgamma_d$-group $G^{*} \leq \Sym(\Omega^{*})$,
 two strings $\mathfrak{x}^{*},\mathfrak{y}^{*}\colon\Omega^{*}\rightarrow\Sigma$, and a sequence of partitions
 $\{\Omega^{*}\} = \mathfrak{B}_0^{*} \succ \dots \succ \mathfrak{B}_k^{*} = \{\{\alpha^{*}\} \mid \alpha^{*} \in \Omega^{*}\}$ of the set $\Omega^{*}$
 such that the following holds:
 \begin{enumerate}
  \item\label{item:reduction-step-one-1} $|\Omega^{*}| \leq n^{c_1 \log d + c_2 + 1}$ for some absolute
    constants $c_1,c_2$ where $n = |\Omega|$,
  \item $G^{*}$ is transitive,
  \item $\mathfrak{B}_i^{*}$ is $G^{*}$-invariant for all $i \in [k]$,
  \item $\mathfrak{x} \cong_G \mathfrak{y}$ if and only if $\mathfrak{x}^{*} \cong_{G^{*}} \mathfrak{y}^{*}$, and
  \item\label{item:reduction-step-one-5} for every $i \in [k]$ and $B \in \mathfrak{B}_{i-1}^{*}$ it holds that
   \begin{enumerate}
    \item $\left(G^{*}\right)_B^{\mathfrak{B}_{i}^{*}[B]}$ is semi-regular, or
    \item $\left(G^{*}\right)_B^{\mathfrak{B}_{i}^{*}[B]}$ is permutationally equivalent to $A_m^{(t)}$ for some $m \leq d$ and $t \leq \frac{m}{2}$  where $m > 4 \log s$ for $s = \binom{m}{t}$.
   \end{enumerate}  
 \end{enumerate}
 Moreover, one can compute all objects in time polynomial in the size of $\Omega^{*}$.
\end{theorem}

We remark that the constants required for Property \ref{item:reduction-step-one-1} are precisely the constants
from Theorem \ref{thm:first-main-theorem}, Option \ref{item:first-main-theorem-1}.

\begin{proof}
 Inductively, by changing the action of~$G$, we transform a sequence 
 $\{\Omega\} = \mathfrak{B}_0 \succ \dots \succ \mathfrak{B}_\ell$ for which the last partition is not discrete
 into a sequence $\{\Omega^*\} = \mathfrak{B}^*_0 \succ \dots \succ \mathfrak{B}^*_{k}$
 such that the blocks in~$\mathfrak{B}^*_{k}$ are smaller than the blocks in~$\mathfrak{B}_\ell$.

 More precisely, let $\{\Omega\} = \mathfrak{B}_0 \succ \dots \succ \mathfrak{B}_\ell$
 be sequence of $G$-invariant partitions such that Property
 \ref{item:reduction-step-one-5} holds for all $i \leq \ell$ with respect to the group $G$,
 i.e.\ for all $i \in [\ell]$ and $B \in \mathfrak{B}_{i-1}$ it holds that
 \begin{enumerate}
  \item[(A)] $G_B^{\mathfrak{B}_{i}[B]}$ is semi-regular, or
  \item[(B)] $G_B^{\mathfrak{B}_{i}[B]}$ is permutationally equivalent to $A_m^{(t)}$ for some $m \leq d$ and $t \leq \frac{m}{2}$  where $m > 4 \log s$ for $s = \binom{m}{t}$.
 \end{enumerate}  

 Let $q = |B|$ for some (and therefore every) $B \in \mathfrak{B}_\ell$ and suppose $q > 1$.
 We show that there is a set $\Omega^{*}$, a $\ourgamma_d$-group $G^{*} \leq \Sym(\Omega^{*})$,
 two strings $\mathfrak{x}^{*},\mathfrak{y}^{*}\colon\Omega^{*}\rightarrow\Sigma$ and a sequence of partitions
 $\{\Omega^{*}\} = \mathfrak{B}_0^{*} \succ \dots \succ \mathfrak{B}_k^{*}$ of the set $\Omega^{*}$
 and natural numbers $b,p$ with $b > 1$ and $b \cdot p = q$ such that the following holds:
 \begin{enumerate}
  \item[(i)]\label{item:i} $|\Omega^{*}| \leq n \cdot b^{c_1 \log d + c_2}$,
  \item[(ii)] $G^{*}$ is transitive,
  \item[(iii)] $\mathfrak{B}_i^{*}$ is $G^{*}$-invariant for all $i \in [k]$,
  \item[(iv)] $|B| \leq p$ for all $B \in \mathfrak{B}_k^{*}$,
  \item[(v)] $\mathfrak{x} \cong_G \mathfrak{y}$ if and only if $\mathfrak{x}^{*} \cong_{G^{*}} \mathfrak{y}^{*}$, and
  \item[(vi)]\label{item:vi} for every $i \in [k]$ and $B \in \mathfrak{B}_{i-1}^{*}$ it holds that
   \begin{enumerate}
    \item[(a)] $\left(G^{*}\right)_B^{\mathfrak{B}_{i}^{*}[B]}$ is semi-regular, or
    \item[(b)] $\left(G^{*}\right)_B^{\mathfrak{B}_{i}^{*}[B]}$ is permutationally equivalent to $A_m^{(t)}$ for some $m \leq d$ and $t \leq \frac{m}{2}$  where $m > 4 \log s$ for $s = \binom{m}{t}$.
   \end{enumerate}  
 \end{enumerate}
 Moreover, one can compute all objects in time polynomial in the size of $\Omega^{*}$.
 Then the statement of the theorem follows by a simple induction.
 
 Let $\mathfrak{B}_{\ell+1} \prec \mathfrak{B}_\ell$ be a $G$-invariant partition such that $G_B^{\mathfrak{B}_{\ell+1}[B]}$ is primitive for every $B \in \mathfrak{B}_\ell$.
 Note that such a partition can be computed in polynomial time by computing a maximal block $B'$ of the group $G_B^{B}$ (where $B \in \mathfrak{B}_{\ell}$ is arbitrary) and setting $\mathfrak{B}_{\ell+1} = \{(B')^{g} \mid g \in G\}$.
 Let $b = |\mathfrak{B}_{\ell+1}[B]|$ for some $B \in \mathfrak{B}_\ell$ and $p = q/b =|B'|$ for $B' \in \mathfrak{B}_{\ell+1}$.
 For $B \in \mathfrak{B}_\ell$ let $G^{B} = G_B^{\mathfrak{B}_{\ell+1}[B]}$ and let
 \begin{equation}
  N^{B} = \begin{cases}
           \Soc\left(G^{B}\right) &\text{if } \left|G^{B}\right| > b^{c_1 \log d + c_2} \text{, and}\\
           \{1\}                  &\text{otherwise.}
          \end{cases}
 \end{equation}
 Note that $N^{B} \unlhd G^{B}$.
 By Theorem \ref{thm:first-main-theorem} it holds that $|G^{B} :
 N^{B}| \leq b^{c_1 \log d + c_2}$.
 As described above the main idea is now to act on the set of cosets of $N^{B}$ in $G^{B}$.
 The main difficulty in defining this action is that there are multiple blocks each equipped with a set of cosets on which the action needs to be defined in a consistent manner. 
 Let \[\Omega^{*} = \bigcup_{B \in \mathfrak{B}_\ell} \big(\{N^{B}h \mid h \in G^{B}\} \times B\big).\]
 Note that $|\Omega^{*}| \leq n \cdot b^{c_1 \log d + c_2}$ and hence Property (i) holds.
 In order to define the action of $G$ on the set $\Omega^{*}$ we first fix a set of elements mapping the blocks in $\mathfrak{B}_\ell$ onto each other.
 Suppose $\mathfrak{B}_\ell = \{B_1,\dots,B_s\}$.
 For $i = 2,\dots,s$ let $\sigma_{1\rightarrow i} \in G$ such that $B_1^{\sigma_{1 \rightarrow i}} = B_i$ and let $\sigma_{1 \rightarrow 1} = 1$ (the identity element).
 Since the groups $G^{B}$ are defined on the domain $\mathfrak{B}_{\ell+1}[B]$ it is actually convenient to define group elements $\overline{\sigma_{1\rightarrow i}}$ over the same domain.
 Let \[\overline{\sigma_{1 \rightarrow i}} \coloneqq \left((\sigma_{1 \rightarrow i})^{\mathfrak{B}_{\ell+1}}\right)|_{\mathfrak{B}_{\ell+1}[B_1]}\]
 (that is, we consider the natural action of $\sigma_{1\rightarrow i}$ on $\mathfrak{B}_{\ell+1}$ and restrict the domain to the set $\mathfrak{B}_{\ell+1}[B_1]$).
 Note that the image of $\overline{\sigma_{1 \rightarrow i}}$ is precisely $\mathfrak{B}_{\ell+1}[B_i]$.
 Moreover, for $i,j \in [k]$ let \[\overline{\sigma_{i \rightarrow j}} \coloneqq \overline{\sigma_{1 \rightarrow i}}^{\,-1}\overline{\sigma_{1 \rightarrow j}}.\]
 Note that $\overline{\sigma_{i \rightarrow j}}^{-1} = \overline{\sigma_{j \rightarrow i}}$ and $\overline{\sigma_{i \rightarrow j}}\,\overline{\sigma_{j \rightarrow r}} = \overline{\sigma_{i \rightarrow r}}$.
 Additionally, for every element $g \in G$ we introduce a similar notation defining
 \[\overline{g_{(i)}} =
   \left(g^{\mathfrak{B}_{\ell+1}}\right)|_{\mathfrak{B}_{\ell+1}[B_i]}.\]
 Now the group $G$ acts on the set $\Omega^{*}$ via
 \[
   (N^{B_i}h,\alpha)^{g} \coloneqq
   ((N^{B_i}h)^{g},\alpha^{g})
 \]
 where $\alpha \in B_i$ and, choosing $j$ such that $B_j = B_i^{g}$,
 \[(N^{B_i}h)^{g} \coloneqq \overline{\sigma_{i\rightarrow j}}^{\,-1}(N^{B_i}h) \overline{g_{(i)}}\]
 To argue that this is well-defined we need to argue that~$(N^{B_i}h,\alpha)^{g} \in \Omega^*$. For this first note that
 \begin{align*}
  \overline{\sigma_{i\rightarrow j}}^{\,-1}(N^{B_i}h)\overline{g_{(i)}}
  &= \overline{\sigma_{i\rightarrow j}}^{\,-1}N^{B_i}\overline{\sigma_{i\rightarrow j}}\,\overline{\sigma_{i\rightarrow j}}^{\,-1}h\overline{g_{(i)}}\\
  &= N^{B_j}\overline{\sigma_{i\rightarrow j}}^{\,-1}h\overline{g_{(i)}}.
 \end{align*}
 Observe that $\overline{\sigma_{i\rightarrow j}}^{\,-1}h\overline{g_{(i)}} \in G^{B_j}$ and $\alpha^{g} \in B_j$. It follows that~$(N^{B_i}h,\alpha)^{g} \in \Omega^*$. 
 
 We also need to argue that we really defined an action of~$G$ on~$\Omega^*$.
 For this let $g,g' \in G$ and $\alpha \in B_i$.
 Pick $j,r \in [s]$ such that $B_i^{g} = B_j$ and $B_j^{g'} = B_r$.
 Note that $B_r = B_i^{(gg')}$.
 Then
 \begin{align*}
  \left(N^{B_i}h,\alpha\right)^{(gg')} &= \left(\overline{\sigma_{i\rightarrow r}}^{\,-1}\left(N^{B_i}h\right)\overline{(gg')_{(i)}},\alpha^{(gg')}\right)\\
                                       &= \left(\overline{\sigma_{j\rightarrow r}}^{\,-1}\overline{\sigma_{i\rightarrow j}}^{\,-1}\left(N^{B_i}h\right)\overline{g_{(i)}}\,\overline{g'_{(j)}},\left(\alpha^{g}\right)^{g'}\right)\\
                                       &= \left(\overline{\sigma_{i\rightarrow j}}^{\,-1}\left(N^{B_i}h\right)\overline{g_{(i)}},\alpha^{g}\right)^{g'}\\
                                       &= \left(\left(N^{B_i}h,\alpha\right)^{g}\right)^{g'}.
 \end{align*}
 Now let $G^{*} = G^{\Omega^{*}}$ be the induced action of $G$ on $\Omega^{*}$ (at this point $G^{*}$ may not be transitive, this is fixed later).
 Also, for $g \in G$, let $g^{*} = g^{\Omega^{*}}$.
 For a string $\mathfrak{x}\colon\Omega \rightarrow \Sigma$ we define the string $\mathfrak{x}^{*}\colon \Omega^{*}\rightarrow\Sigma\colon(N^{B}h,\alpha) \mapsto \mathfrak{x}(\alpha)$.
 
 \begin{claim}
  \label{claim:reduction-one-1}
  For every $g \in G$ it holds that $\mathfrak{x}^{g} = \mathfrak{y}$ if and only if $\left(\mathfrak{x}^{*}\right)^{g^{*}} = \mathfrak{y}^{*}$.
  \proof
  Let $g\in G$. For the forward direction, 
  suppose that $\mathfrak{x}^{g} = \mathfrak{y}$, that is, $\mathfrak{x}(\alpha) = \mathfrak{y}(\alpha^{g})$ for all $\alpha \in \Omega$.
  Then $\mathfrak{x}^{*}(N^{B}h,\alpha) = \mathfrak{x}(\alpha) = \mathfrak{y}(\alpha^{g}) =
  \mathfrak{y}^{*}((N^{B}h)^{g},\alpha^{g}) =
  \mathfrak{y}^{*}((N^{B}h,\alpha)^{g^{*}})$ and hence,
  $(\mathfrak{x}^{*})^{g^{*}} =
  \mathfrak{y}^{*}$.
  
  For the backward direction, suppose that $(\mathfrak{x}^{*})^{g^{*}} = \mathfrak{y}^{*}$.
  Let $\alpha \in \Omega$ and let $B \in \mathfrak{B}_\ell$ such that $\alpha \in B$ and let $h \in G^{B}$.
  Then $\mathfrak{x}(\alpha) = \mathfrak{x}^{*}(N^{B}h,\alpha) = \mathfrak{y}^{*}((N^{B}h,\alpha)^{g^{*}})
  = \mathfrak{y}^{*}((N^{B}h)^{g},\alpha^{g}) = \mathfrak{y}(\alpha^{g})$.
  So $\mathfrak{x}^{g} = \mathfrak{y}$.
  \uend
 \end{claim}
 Next, we define the desired sequence of partitions.
 For $i \in [\ell]$ let
 \[\mathfrak{B}_i^{*} = \left\{\bigcup_{B \in \mathfrak{B}_\ell\colon B \subseteq B'} \left\{N^{B}h \mid h \in G^{B}\right\} \times B \;\middle|\; B' \in \mathfrak{B}_i\right\}.\]
 It is easy to check that $\mathfrak{B}_i^{*}$ is $G^{*}$-invariant for all $i \in [\ell]$.
 Moreover, $G^{\mathfrak{B}_\ell}$ is permutationally equivalent to
 $\left(G^{*}\right)^{\mathfrak{B}_\ell^{*}}$ via the permutational isomorphism
 $f\colon \mathfrak{B}_\ell \rightarrow \mathfrak{B}_\ell^{*}$ where
 \[f(B) = \{N^{B}h \mid h \in G^{B}\} \times B\]
 for all $B \in \mathfrak{B}_\ell$.
 As a result $G^{\mathfrak{B}_i}$ is permutationally equivalent to
 $\left(G^{*}\right)^{\mathfrak{B}_i^{*}}$ for all $i \in [\ell]$.
 In particular, Property (vi) holds for all $i \in [\ell]$.
 
 Now we distinguish two cases. First suppose $\left|G^{B}\right| \leq b^{c_1 \log d + c_2}$ where $b = |\mathfrak{B}_{\ell+1}[B]|$ for some (and therefore every) $B \in \mathfrak{B}_\ell$.
 Recall that $N^{B} = \{1\}$ in this case.
 Let
 \[\mathfrak{B}_{\ell+1}^{*} = \left\{\{N^{B}h\} \times B' \mid h \in G^{B}, B \in \mathfrak{B}_\ell, B' \in \mathfrak{B}_{\ell+1} \text{ with } B' \subseteq B\right\}\]
 and set $k = \ell + 1$.
 Clearly, $\mathfrak{B}_{\ell+1}^{*}$ is $G^{*}$-invariant and $|B^*| \leq p$ for all $B^* \in \mathfrak{B}_{\ell+1}^{*}$.
 Now consider the group
 \[\left(G^{*}\right)^{B^*} = \left(G^{*}\right)_{B^*}^{\mathfrak{B}_{\ell+1}^{*}[B^*]}\]
 for $B^* \in \mathfrak{B}_\ell^{*}$.
 It is easy to check that $\left(G^{*}\right)^{B^*}$ is permutationally equivalent to $G^{B} = G_{B}^{\mathfrak{B}_{\ell+1}[B]}$
 with its natural action on the set $G^{B} \times \mathfrak{B}_{\ell+1}[B]$ (acting regularly on the first component) where $B^{*} = \left\{N^{B}h \mid h \in G^{B}\right\} \times B$ and $B \in \mathfrak{B}_\ell$.
 Hence, $\left(G^{*}\right)^{B^*}$ is semi-regular.
 So it only remains to argue that (ii) holds.
 Indeed, the group $G^{*}$ is not necessarily transitive.
 Let $A^{*} \subseteq \Omega^{*}$ be an orbit of $G^{*}$.
 Then, by restricting all partitions and the two strings to $A^{*}$ the group $\left(G^{*}\right)^{A^*}$ satisfies all required properties.
 This is trivial for all properties but (v).
 For Property (v) note that $G$ is transitive.
 So for every $\alpha \in \Omega$ there is some element $a \in A^{*}$ whose second component is $\alpha$.
 This is all we need to prove a variant of Claim~\ref{claim:reduction-one-1} where we restrict the strings and the group to the set $A^{*}$.
 
 In the other case $\left|G^{B}\right| > b^{c_1 \log d + c_2}$ and $N^{B} = \Soc(G^{B})$.
 We consider the block $B_1 \in \mathfrak{B}_\ell$ (recall that in the beginning of the proof we fixed a numbering of the blocks in $\mathfrak{B}_\ell$ and elements $\sigma_{1 \rightarrow i}$ mapping the first block to the $i$-th block).
 By Theorem \ref{thm:first-main-theorem} there is a sequence of partitions
 $\{\mathfrak{B}_{\ell+1}[B_1]\} = \mathcal{P}_0 \succ \dots \succ \mathcal{P}_t = \{\{B'\} \mid B' \in \mathfrak{B}_{\ell+1}[B_1]\}$ such that
 \begin{enumerate}
  \item[(I)] $\mathcal{P}_i$ is $N^{B_1}$-invariant for every $i \in [t]$, and
  \item[(II)] there are $m \leq d$ and $t \leq \frac{m}{2}$ with $m > 4 \log s$ where $s = \binom{m}{t}$ such that for all $i \in [t]$ and $P \in \mathcal{P}_{i-1}$
        the group $\left(N^{B_1}\right)_P^{\mathcal{P}_{i}[P]}$ is permutationally equivalent to $A_m^{(t)}$.
 \end{enumerate}
 Let
 \[\mathfrak{B}_{\ell+1}^{*} = \left\{\{N^{B}h\} \times B \mid h \in G^{B}, B \in \mathfrak{B}_\ell\right\}\]
 and for $i \in [t]$ let
 \[\mathfrak{B}_{\ell+1+i}^{*} = \left\{\{N^{B_j}h\} \times
     \left(\bigcup_{B' \in P} (B')^{\sigma_{1\rightarrow j}}\right)
     \middle| h \in G^{B_j}, j \in [s], P \in \mathcal{P}_i\right\}.\]
 We set $k = \ell + 1 + t$.
 First note that $|B^{*}| \leq p$ for every $B^{*} \in \mathfrak{B}_k^{*}$.
 We argue that $\mathfrak{B}_{\ell+1+i}^{*}$ is a $G^{*}$-invariant partition for $i \in [t]$.
 Let $B^{*} \in \mathfrak{B}_{\ell+1+i}^{*}$ and $g^{*} \in G^{*}$ such that $(B^{*})^{g^{*}} \cap B^{*} \neq \emptyset$.
 Let $g \in G$ be the element corresponding to $g^*$ and suppose
 \[B^*=\{N^{B_j}h\} \times \left(\bigcup_{B' \in P} (B')^{\sigma_{1\rightarrow j}}\right),\]
 and let $B=B_j$
 Due to the action on the first component (i.e.\ the action on $N^Bh$)
 we conclude that $g \in G_B$ and $g^{\mathfrak{B}_{\ell+1}[B]} \in N^{B}$.
 Since $\mathcal{P}_i$ is $N^{B_1}$-invariant we conclude
 that $(B^{*})^{g^{*}} = B^{*}$.
 Hence, $\mathfrak{B}_{\ell+1+i}^{*}$ is a $G^{*}$-invariant partition.
 
 Moreover, $(G^{*})_{B^{*}}^{\mathfrak{B}^{*}_{\ell+1}[B^{*}]}$ is semi-regular for every $B^{*} \in \mathfrak{B}^{*}_\ell$
 and, for every $i \in [t]$ and every $B^{*} \in
 \mathfrak{B}^{*}_{\ell+i}$, the group
 $(G^{*})_{B^{*}}^{\mathfrak{B}^{*}_{\ell+1+i}[B^{*}]}$ is
 permutationally equivalent to
 $\left(N^{B_1}\right)_P^{\mathcal{P}_{i}[P]}$.
 To see this first observe that $\mathfrak{B}_{\ell+1+t}^{*} = \{\{N^{B}h\}\times B' \mid B \in \mathfrak{B}_\ell, B' \in \mathfrak{B}_{\ell+1},B' \subseteq B, h \in G^{B}\}$.
 Let $B^{*} \in \mathfrak{B}_{\ell+1}^{*}$ and suppose $B^{*} = \{N^{B}h\} \times B$ where $B \in \mathfrak{B}_\ell$.
 For every $g^{*} \in (G^{*})_{B^{*}}$ it holds that $(N^{B}h)^{g} = N^{B}h$ where $g$ is the element corresponding to $g^{*}$ and thus, $g^{\mathfrak{B}_{\ell+1}[B]} \in N^{B}$.
 Hence, $(G^{*})_{B^{*}}^{\mathfrak{B}_{\ell+1+t}^{*}[B^{*}]}$ is permutationally equivalent to $N^{B}$.
 Translating the sequence of partitions $\mathcal{P}_1,\dots,\mathcal{P}_t$ for $N^{B}$ back then gives the sequence of partitions described above.
 
 Finally, as in the previous case, if the group $G^{*}$ is not transitive we restrict the group (along with strings $\mathfrak{x}^{*}$ and $\mathfrak{y}^{*}$) to one of its orbits.
\end{proof}

\subsection{Second Step}

For the second step we require the following auxiliary lemmata. The first one is implicitly given in \cite[Section 4]{BLS87}.

\begin{lemma}[\cite{BLS87}]
 \label{la:compute-johnson-representation}
 Let $G \leq \Sym(\Omega)$ and suppose $G$ is permutationally equivalent to $A_{m}^{(t)}$ or $S_m^{(t)}$ for $m > 4 \log n$ and $t \leq \frac{m}{2}$.
 Then a permutational isomorphism $\rho\colon\Omega \rightarrow \binom{[m]}{t}$
 can be computed in polynomial time.
\end{lemma}

\begin{lemma}
 \label{la:permutational-automorphisms-johnson-action}
 Let $m \geq 7$ and suppose $\gamma \in \Sym(\binom{[m]}{t})$ is a permutational automorphism of $A_m^{(t)}$.
 Then $\gamma$ is induced by a unique permutation $\sigma \in S_m$, that is, $X^{\gamma} = X^{\sigma} = \{x^{\sigma} \mid x \in X\}$ for every $X \in \binom{[m]}{t}$.
\end{lemma}

\begin{proof}
 Every non-trivial permutational automorphism of $A_m^{(t)}$ gives a unique non-trivial automorphism of $A_m$ (as an abstract group) and every element $\sigma \in S_m$ induces a permutational automorphism of $A_m^{(t)}$.
 Since $\Aut(A_m) \cong S_m$ (for $m \geq 7$) the statement follows.
\end{proof}

\begin{theorem}
 \label{thm:reduction-two}
 Let $G \leq \Sym(\Omega)$ be a transitive $\ourgamma_d$-group and let $\mathfrak{x},\mathfrak{y}\colon\Omega\rightarrow\Sigma$ be two strings.
 Then there is a set $\Omega^{*}$, a $\ourgamma_d$-group $G^{*} \leq \Sym(\Omega^{*})$,
 two strings
 $\mathfrak{x}^{*},\mathfrak{y}^{*}\colon\Omega^{*}\rightarrow\Sigma$
 and a~$G^{*}$-invariant almost \emph{$d$-ary} sequence of partitions
 $\{\Omega^{*}\} = \mathfrak{B}_0^{*} \succ \dots \succ \mathfrak{B}_k^{*} = \{\{\alpha^{*}\} \mid \alpha^{*} \in \Omega^{*}\}$ of the set $\Omega^{*}$
 such that the following holds:
 \begin{enumerate}
  \item\label{item:reduction-step-two-1} $|\Omega^{*}| \leq n^{(c_1 \log d + c_2 + 1)\log d}$, and
  \item $\mathfrak{x} \cong_G \mathfrak{y}$ if and only if $\mathfrak{x}^{*} \cong_{G^{*}} \mathfrak{y}^{*}$.
 \end{enumerate}
 Moreover, one can compute all objects in time polynomial in the size of $\Omega^{*}$.
\end{theorem}

\begin{proof}
 By Theorem \ref{thm:reduction-one} we can assume that there is a sequence of $G$-invariant partitions $\{\Omega\} = \mathfrak{B}_0 \succ \dots \succ \mathfrak{B}_\ell = \{\{\alpha\} \mid \alpha \in \Omega\}$
 such that for every $i \in [\ell]$ and $B \in \mathfrak{B}_{i-1}$ it holds that
 \begin{enumerate}
  \item[(A)] $G_B^{\mathfrak{B}_{i}[B]}$ is semi-regular, or
  \item[(B)] $G_B^{\mathfrak{B}_{i}[B]}$ is permutationally equivalent to $A_m^{(t)}$ for some $m \leq d$ and $t \leq \frac{m}{2}$ where $m > 4 \log s$ for $s = \binom{m}{t}$.
 \end{enumerate}
 (Actually, using Theorem \ref{thm:reduction-one}, the above condition can only be achieved by increasing the size of the set $\Omega$ as described in Theorem \ref{thm:reduction-one}, Property \ref{item:reduction-step-one-1}.
 We argue that under the above assumption the set $\Omega^{*}$ constructed in this proof
 has size at most $n^{\log d}$ which in combination with Theorem
 \ref{thm:reduction-one} results in the desired bound given in
 \ref{item:reduction-step-two-1}.)
 
 In order to get almost~$d$-arity,
 we need to worry about those blocks that satisfy item (B).
 Let \[I = \left\{i \in [\ell] \mid \exists B \in \mathfrak{B}_{i-1} \colon G_B^{\mathfrak{B}_{i}[B]} \text{ is permutationally equivalent to } A_{m_i}^{(t_i)}\right\}.\]
 Note that for $B,B' \in \mathfrak{B}_{i-1}$ the groups $G_B^{\mathfrak{B}_{i}[B]}$ and $G_{B'}^{\mathfrak{B}_{i}[B']}$ are permutationally equivalent.
 So the existential quantifier in the definition of the set $I$ can also be replaced by a universal quantifier.
 
 For $i \in I$ and $B \in \mathfrak{B}_{i-1}$ let $\rho_{i,B}\colon \mathfrak{B}_{i}[B] \rightarrow \binom{[m_i]}{t_i}$ be a permutational isomorphism from $G_B^{\mathfrak{B}_{i}[B]}$ to $A_{m_i}^{(t_i)}$.
 Note that such a $\rho_{i,B}$ can be computed in polynomial time by Lemma \ref{la:compute-johnson-representation}.
 Let $\Gamma = (V(\Gamma),E(\Gamma))$ be the directed graph with
 \[
  V(\Gamma) =\bigcup_{i \in \{0,\dots,\ell\}} \mathfrak{B}_i
              \cup\left\{(i,B,X) \mid i \in I, B \in \mathfrak{B}_{i-1}, X \in \binom{[m_i]}{\leq t_i}\right\}
 \]
 and
 \begin{align*}
  \left((i,B,X),(i',B',X')\right) \in E(\Gamma) \;\;\; :\Leftrightarrow \;\;\; &i=i' \wedge B = B' \wedge X \subseteq X' \wedge |X' \setminus X| = 1 ,\\
  \left(B,(i,B',X)\right) \in E(\Gamma)         \;\;\; :\Leftrightarrow \;\;\; &B = B' \wedge X = \emptyset ,\\
  \left((i,B,X),B'\right) \in E(\Gamma)         \;\;\; :\Leftrightarrow \;\;\; &|X| = t_i \wedge B' \in \mathfrak{B}_i \wedge B' \subseteq B \wedge \rho_{i,B}(B') = X,\\
  \left(B,B'\right) \in E(\Gamma)               \;\;\; :\Leftrightarrow \;\;\; & \exists i\in [\ell]\setminus I \colon B \in \mathfrak{B}_{i-1} \wedge B' \in \mathfrak{B}_i \wedge B' \subseteq B.
 \end{align*}
 We regard $v_0 \coloneqq \Omega\in\mathfrak B_0$ as the ``root'' of $\Gamma$ ($v_0$ is the unique vertex with in-degree $0$ in $\Gamma$).
 An illustration of the graph $\Gamma$ is given in Figure~\ref{fig:visualization-gamma}.
 
 A \emph{branch of $(\Gamma,v_0)$} is a path $(v_0,v_1,\dots,v_p)$ such that $\dist(v_0,v_i) = i$ for all $i \in [p]$.
 A \emph{maximal branch of $(\Gamma,v_0)$} is a branch of maximal length.
 Observe that for every maximal branch $(v_0,v_1,\dots,v_p)$ it holds that $v_p = \{\alpha\}$ for some $\alpha \in \Omega$.
 Let $M$ be the set of maximal branches of $(\Gamma,v_0)$.
 
 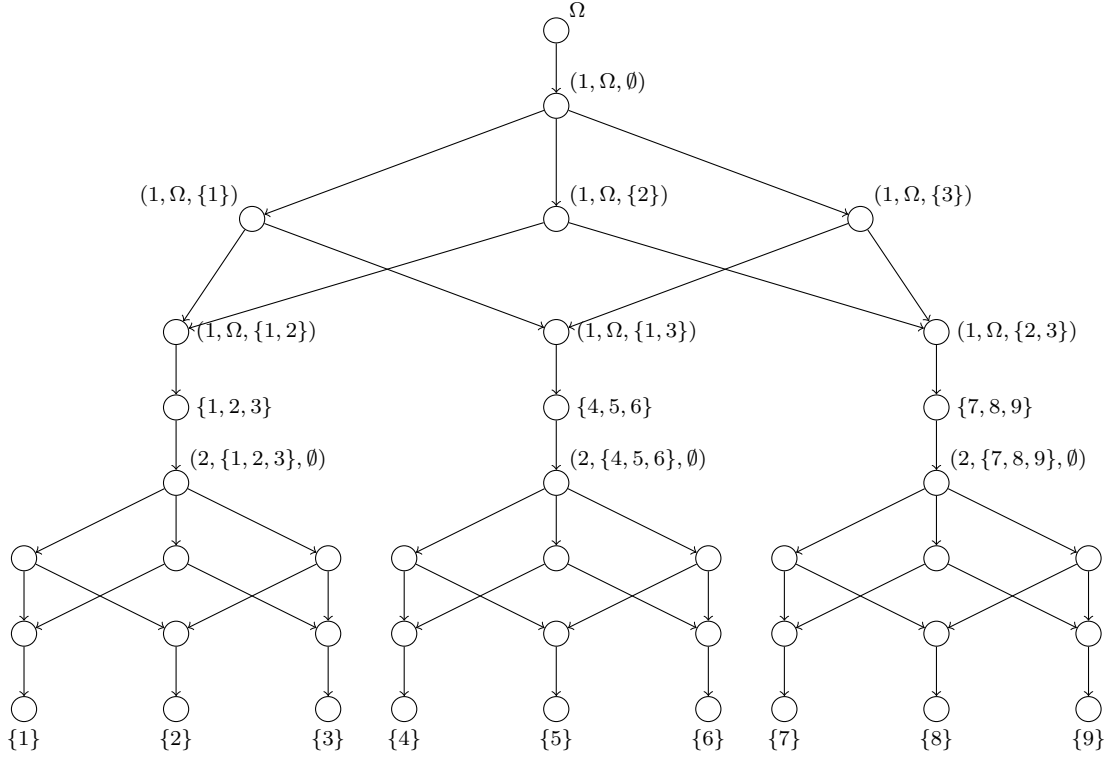
\begin{figure}
  \centering
  \begin{tikzpicture}
   \node[vertex,label={[label distance=-3]45:\footnotesize{$\Omega$}}] (v00) at (0,1.0) {};
   
   \node[vertex,label={[label distance=-3]45:\footnotesize{$(1,\Omega,\emptyset)$}}] (v0) at (0,0) {};
   \node[vertex,label={[label distance=-3]135:\footnotesize{$(1,\Omega,\{1\})$}}] (v1) at (-4,-1.5) {};
   \node[vertex,label={[label distance=-3]45:\footnotesize{$(1,\Omega,\{2\})$}}] (v2) at (0,-1.5) {};
   \node[vertex,label={[label distance=-3]45:\footnotesize{$(1,\Omega,\{3\})$}}] (v3) at (4,-1.5) {};
   \node[vertex,label={[label distance=-1]0:\footnotesize{$(1,\Omega,\{1,2\})$}}] (v12) at (-5,-3) {};
   \node[vertex,label={[label distance=-1]0:\footnotesize{$(1,\Omega,\{1,3\})$}}] (v13) at (0,-3) {};
   \node[vertex,label={[label distance=-1]0:\footnotesize{$(1,\Omega,\{2,3\})$}}] (v23) at (5,-3) {};
   
   \node[vertex,label={[label distance=-1]0:\footnotesize{$\{1,2,3\}$}}] (w12) at (-5,-4) {};
   \node[vertex,label={[label distance=-1]0:\footnotesize{$\{4,5,6\}$}}] (w13) at (0,-4) {};
   \node[vertex,label={[label distance=-1]0:\footnotesize{$\{7,8,9\}$}}] (w23) at (5,-4) {};
   
   \node[vertex,label={[label distance=-3]45:\footnotesize{$(2,\{1,2,3\},\emptyset)$}}] (x0-1) at ($(-5,-5)+(0,0)$) {};
   \node[vertex,label={[label distance=-3]45:\footnotesize{$(2,\{4,5,6\},\emptyset)$}}] (x0-2) at ($(0,-5)+(0,0)$) {};
   \node[vertex,label={[label distance=-3]45:\footnotesize{$(2,\{7,8,9\},\emptyset)$}}] (x0-3) at ($(5,-5)+(0,0)$) {};
   
   \foreach \i in {1,2,3}{
    \node[vertex] (x1-\i) at ($(-10+5*\i,-5)+(-2,-1)$) {};
    \node[vertex] (x2-\i) at ($(-10+5*\i,-5)+(0,-1)$) {};
    \node[vertex] (x3-\i) at ($(-10+5*\i,-5)+(2,-1)$) {};
    \node[vertex] (x12-\i) at ($(-10+5*\i,-5)+(-2,-2)$) {};
    \node[vertex] (x13-\i) at ($(-10+5*\i,-5)+(0,-2)$) {};
    \node[vertex] (x23-\i) at ($(-10+5*\i,-5)+(2,-2)$) {};
   }
   
   \node[vertex] (y12-1) at ($(-5,-5)+(-2,-3)$) {};
   \node[vertex] (y13-1) at ($(-5,-5)+(0,-3)$) {};
   \node[vertex] (y23-1) at ($(-5,-5)+(2,-3)$) {};
   \node[vertex] (y12-2) at ($(0,-5)+(-2,-3)$) {};
   \node[vertex] (y13-2) at ($(0,-5)+(0,-3)$) {};
   \node[vertex] (y23-2) at ($(0,-5)+(2,-3)$) {};
   \node[vertex] (y12-3) at ($(5,-5)+(-2,-3)$) {};
   \node[vertex] (y13-3) at ($(5,-5)+(0,-3)$) {};
   \node[vertex] (y23-3) at ($(5,-5)+(2,-3)$) {};
   
   \node at ($(-5,-5)+(-2,-3)+(270:0.4)$) {\footnotesize{$\{1\}$}};
   \node at ($(-5,-5)+(0,-3)+(270:0.4)$) {\footnotesize{$\{2\}$}};
   \node at ($(-5,-5)+(2,-3)+(270:0.4)$) {\footnotesize{$\{3\}$}};
   \node at ($(0,-5)+(-2,-3)+(270:0.4)$) {\footnotesize{$\{4\}$}};
   \node at ($(0,-5)+(0,-3)+(270:0.4)$) {\footnotesize{$\{5\}$}};
   \node at ($(0,-5)+(2,-3)+(270:0.4)$) {\footnotesize{$\{6\}$}};
   \node at ($(5,-5)+(-2,-3)+(270:0.4)$) {\footnotesize{$\{7\}$}};
   \node at ($(5,-5)+(0,-3)+(270:0.4)$) {\footnotesize{$\{8\}$}};
   \node at ($(5,-5)+(2,-3)+(270:0.4)$) {\footnotesize{$\{9\}$}};
   
   \path[draw,->]
    (v00) edge (v0)
    (v0) edge (v1)
    (v0) edge (v2)
    (v0) edge (v3)
    (v1) edge (v12)
    (v1) edge (v13)
    (v2) edge (v12)
    (v2) edge (v23)
    (v3) edge (v13)
    (v3) edge (v23)
    (v12) edge (w12)
    (v13) edge (w13)
    (v23) edge (w23)
    (w12) edge (x0-1)
    (w13) edge (x0-2)
    (w23) edge (x0-3);
   
   \foreach \i in {1,2,3}{
    \path[draw,->]
     (x0-\i) edge (x1-\i)
     (x0-\i) edge (x2-\i)
     (x0-\i) edge (x3-\i)
     (x1-\i) edge (x12-\i)
     (x1-\i) edge (x13-\i)
     (x2-\i) edge (x12-\i)
     (x2-\i) edge (x23-\i)
     (x3-\i) edge (x13-\i)
     (x3-\i) edge (x23-\i)
     (x12-\i) edge (y12-\i)
     (x13-\i) edge (y13-\i)
     (x23-\i) edge (y23-\i);;
   }
  \end{tikzpicture}
  \caption{A visualization of $\Gamma$. Here $\Omega=[9]$, $k=3$, and $\mathfrak{B}_0 = \{\Omega\}$,
    $\mathfrak{B}_1=\big\{\{1,2,3\}, \{4,5,6\},\{7,8,9\}\big\}$,
    $\mathfrak{B}_3=\big\{\{1\},\ldots,\{9\}\big\}$. Furthermore, $I = \{1,2\}$, $m_i=3$ and
    $t_i=2$; we ignore the condition $t_i\le m_i/2$ for illustration purposes.}
  \label{fig:visualization-gamma}
 \end{figure}
 
 \begin{claim}
  \label{claim:second-main-theorem-size}
  $|M| \leq n^{\log d}$.
  \proof
  We can view the sequence of partitions $\mathfrak{B}_i$ as a
  tree of height $\ell$. Each leaf of this tree corresponds to an element $\alpha \in \Omega$.
  
  The graph $\Gamma$ is
  obtained from the partition tree by squeezing subset-lattices of
  the $(\le t_i)$-element subsets of $[m_i]$ between some internal node of
  the tree and its $\binom{m_i}{t_i}$ children. Counting the number of
  branches in $\Gamma$ amounts to counting the number of leaves in the
  tree unfolding of $\Gamma$. To obtain the tree unfolding, we replace
  each of the subset lattices of size $\binom{m_i}{t_i}$ by a tree of size
  $m_i^{t_i}$. For a fixed subset lattice every element $X \subseteq [m_i]$
  of size $t_i$ corresponds to $m_i^{t_i}/\binom{m_i}{t_i}$ many tuples in the tree unfolding.
  Hence,
  \begin{align*}
   |M| &\leq n \cdot \prod_{i \in I} \left(m_i^{t_i}/\binom{m_i}{t_i}\right)\\
       &\leq n \cdot \prod_{i \in I}\binom{m_i}{t_i}^{\log d - 1}
&\text{by Lemma \ref{la:approx-binom}}\\
       &\leq n \cdot \left(\prod_{i \in I}\binom{m_i}{t_i}\right)^{\log d - 1}\\
       &\leq n \cdot \left(\prod_{i \in I}|\mathfrak{B}_{i-1} : \mathfrak{B}_i|\right)^{\log d - 1}\\
       &\leq n^{\log d}.
  \end{align*}
  \uend
 \end{claim}
 
 For every branch $\bar v = (v_0,\dots,v_p) \in M$ define $\sigma(\bar v) = \alpha$ for the unique $\alpha \in \Omega$ such that $v_p = \{\alpha\}$.
 Now let $\Omega^{*} = \{(\alpha,\bar v) \mid \alpha \in \Omega, \bar v \in M, \alpha = \sigma(\bar v)\}$.
 Clearly, $|\Omega^{*}| = |M| \leq n^{\log d}$ by Claim
 \ref{claim:second-main-theorem-size}.
 Let $\mathfrak{x}^{*}\colon \Omega^{*} \rightarrow \Sigma\colon (\alpha,\bar v) \mapsto \mathfrak{x}(\alpha)$ and $\mathfrak{y}^{*}\colon \Omega^{*} \rightarrow \Sigma\colon (\alpha,\bar v) \mapsto \mathfrak{y}(\alpha)$.
 
 For $g \in G$ define $g^{\Gamma} \in \Sym(V(\Gamma))$ to be the permutation defined by
 \[B^{\left(g^{\Gamma}\right)} = B^{g}\]
 and
 \[(i,B,X)^{\left(g^{\Gamma}\right)} = (i,B^{g},X')\]
 where $X'$ is defined as follows.
 Let $g^{\mathfrak{B}_{i}[B]} \colon \mathfrak{B}_{i}[B] \rightarrow \mathfrak{B}_{i}[B^{g}]\colon B' \mapsto (B')^{g}$ and define
 \[f \colon \binom{[m_i]}{t_i} \rightarrow \binom{[m_i]}{t_i}\colon Y \mapsto Y^{\rho_{i,B}^{-1} \;\cdot\; g^{\mathfrak{B}_{i}[B]} \;\cdot\; \rho_{i,B^{g}}}\]
 The bijection $f \in \Sym(\binom{[m_i]}{t_i})$ is induced by a unique permutation $\pi \in S_{m_i}$ (see Lemma \ref{la:permutational-automorphisms-johnson-action}).
 Now define $X' = X^{\pi}$.
 
 \begin{claim}[resume]
  For every $g \in G$ we have $g^{\Gamma} \in \Aut(\Gamma,v_0)$.
  \proof
  Clearly, the root of the graph $\Gamma$ is mapped to the root.
  For the edge relation we consider the four types of edges one by one.
  First suppose $(B,B') \in E(\Gamma)$ for some $B \in \mathfrak{B}_{i-1}$, $B' \in \mathfrak{B}_i$ such that $B' \subseteq B$ and $i \notin I$.
  Then $(B,B')^{\left(g^{\Gamma}\right)} = (B^{g},(B')^{g})$ where $B^{g} \in \mathfrak{B}_{i-1}$, $(B')^{g} \in \mathfrak{B}_i$ and $(B')^{g} \subseteq B^{g}$.
  Hence, $(B,B')^{\left(g^{\Gamma}\right)} \in E(\Gamma)$.
  
  Next suppose $((i,B,X),(i',B',X')) \in E(\Gamma)$.
  Then $i=i'$, $B=B'$, $X \subseteq X'$ and $|X' \setminus X| = 1$.
  Let $\pi,\pi' \in S_{m_i}$ such that $(i,B,X)^{\left(g^{\Gamma}\right)} = (i,B^{g},X^{\pi})$ and $(i',B',X')^{\left(g^{\Gamma}\right)} = (i',(B')^{g},(X')^{\pi'})$.
  From $B=B'$ it follows that $\pi = \pi'$.
  Hence, $X^{\pi} \subseteq (X')^{\pi'}$ and $|(X')^{\pi'} \setminus X^{\pi}| = 1$.
  So $((i,B,X)^{\left(g^{\Gamma}\right)},(i',B',X')^{\left(g^{\Gamma}\right)}) \in E(\Gamma)$.
  
  So assume that $(B,(i,B',X)) \in E(\Gamma)$.
  Then $B = B'$ and $X = \emptyset$.
  Let $\pi \in S_{m_i}$ such that $(i,B',X)^{\left(g^{\Gamma}\right)} = (i,(B')^{g},X^{\pi})$.
  But $X^{\pi} = \emptyset^{\pi} = \emptyset$ and $B^{g} = (B')^{g}$.
  This implies that $(B^{\left(g^{\Gamma}\right)},(i,B',X)^{\left(g^{\Gamma}\right)}) \in E(\Gamma)$.
  
  Finally suppose $((i,B,X),B') \in E(\Gamma)$.
  We have $|X| = t_i$, $B' \in \mathfrak{B}_i$, $B' \subseteq B$ and $\rho_{i,B}(B') = X$.
  Let $\pi \in S_{m_i}$ such that $(i,B,X)^{\left(g^{\Gamma}\right)} = (i,(B)^{g},X^{\pi})$.
  It suffices to show that $\rho_{i,B^{g}}((B')^{g}) = X^{\pi}$.
  But this follows directly from the definition of the permutation $g^{\Gamma}$.
  \uend
 \end{claim}
 
 For $g \in G$ define $g^{*} \in \Sym(\Omega^{*})$ via
 \[(\alpha,(v_0,\dots,v_p))^{g^{*}} = \left(\alpha^{g},\left(v_0^{\left(g^{\Gamma}\right)},\dots,v_p^{\left(g^{\Gamma}\right)}\right)\right)\]
 and let $G^{*} = \{g^{*} \mid g \in G\}$.
 
 \begin{claim}[resume]
  \label{claim:second-main-theorem-multiplication}
  For every $g,h \in G$ it holds that $(gh)^{*} = g^{*}h^{*}$.
  \proof
  Looking at the definition it is not difficult to see that $g^{\Gamma}h^{\Gamma} = (gh)^{\Gamma}$ for all $g,h \in G$.
  From this we immediately get that $(gh)^{*} = g^{*}h^{*}$.
  \uend
 \end{claim}
 
 Hence, $G^{*}$ forms a group and the mapping $\varphi\colon G \rightarrow G^{*}\colon g \mapsto g^{*}$ is a group isomorphism.
 So $G^{*} \in \ourgamma_d$ by Lemma \ref{la:gamma-d-closure}.
 
 \begin{claim}[resume]
  \label{claim:second-main-theorem-isomorphism}
  For every $g \in G$ it holds that $\mathfrak{x}^{g} = \mathfrak{y}$ if and only if $\left(\mathfrak{x}^{*}\right)^{g^{*}} = \mathfrak{y}^{*}$.
  \proof
  First suppose there is some $g \in G$ such that $\mathfrak{x}^{g} = \mathfrak{y}$, that is, $\mathfrak{x}(\alpha) = \mathfrak{y}(\alpha^{g})$ for all $\alpha \in \Omega$.
  Then $\mathfrak{x}^{*}(\alpha,\bar v) = \mathfrak{x}(\alpha) = \mathfrak{y}(\alpha^{g}) = \mathfrak{y}^{*}(\alpha^{g},\bar v^{\left(g^{\Gamma}\right)}) = \mathfrak{y}^{*}((\alpha,\bar v)^{g^{*}})$ and hence, $(\mathfrak{x}^{*})^{g^{*}} = \mathfrak{y}^{*}$.
  
  For the backward direction assume there is some $g^{*} \in G^{*}$ such that $(\mathfrak{x}^{*})^{g^{*}} = \mathfrak{y}^{*}$.
  Let $\alpha \in \Omega$ and let $\bar v \in M$ be a maximal branch such that $\alpha = \sigma(\bar v)$.
  Then $\mathfrak{x}(\alpha) = \mathfrak{x}^{*}(\alpha,\bar v) = \mathfrak{y}^{*}((\alpha,\bar v)^{g^{*}}) = \mathfrak{y}^{*}(\alpha^{g},\bar v^{\left(g^{\Gamma}\right)}) = \mathfrak{y}(\alpha^{g})$.
  So $\mathfrak{x}^{g} = \mathfrak{y}$.
  \uend
 \end{claim}
 
 It remains to define the sequence of block systems.
 Let $k = \max_{v\in V(\Gamma)} \dist(v_0,v)$.
 Note that $k$ is the length of every maximal branch $\bar v \in M$.
 For $i \in \{0,\dots,k\}$ define
 \[\mathfrak{B}_i^{*} = \{\{(\alpha,(v_0,\dots,v_k)) \in \Omega^{*} \mid \forall j \leq i\colon v_j = w_j \} \mid (w_0,\dots,w_i) \text{ is a branch of } (\Gamma,v_0)\}.\]
 Clearly, $\mathfrak{B}_i^{*}$ is invariant under $G^{*}$ and $\mathfrak{B}_{i-1}^{*} \succeq \mathfrak{B}_{i}^{*}$ for all $i \in [k]$.
 So it only remains to check that the sequence $\mathfrak{B}_0^{*},\dots,\mathfrak{B}_k^{*}$ is almost $d$-ary.
 For every $B^{*} \in \mathfrak{B}_{i-1}^{*}$, $i \in [k]$, it holds that
 \[|\mathfrak{B}_{i}^{*}[B^{*}]| \leq \max_{v \in V(\Gamma)\colon \dist(v_0,v) = i-1} |N^{+}(v)|\]
 where $N^{+}(v) = \{w \in N(v) \mid \dist(v_0,w) > \dist(v_0,v)\}$.
 Let $i \in [k]$ and $B^{*} \in \mathfrak{B}_{i-1}^{*}$.
 Suppose that $|\mathfrak{B}_i^{*}[B^{*}]| > d$.
 Let $(w_0,\dots,w_{i-1})$ be the branch of $(\Gamma,v_0)$ that corresponds to the block $B^{*}$.
 Then $|N^{+}(w_{i-1})| > d$ and thus, $w_{i-1} = B$ for some $B \in \mathfrak{B}_{j-1}$ and $j \in [\ell]$ where $j \notin I$.
 Moreover, $G_B^{\mathfrak{B}_j[B]}$ is semi-regular since Property (A) holds for all $B \in \mathfrak{B}_{j-1}$ and $j \in [\ell]$ where $j \notin I$.
 But in this case $(G^{*})_{B^{*}}^{\mathfrak{B}_j^{*}[B^{*}]}$ is permutationally equivalent to a subgroup of $G_B^{\mathfrak{B}_j[B]}$ and hence,
 the group $(G^{*})_{B^{*}}^{\mathfrak{B}_j^{*}[B^{*}]}$ is also semi-regular.
\end{proof}

The previous theorem states that there is an $n^{\polylog{d}}$-reduction from the String Isomorphism Problem for $\ourgamma_d$-groups
to the String Isomorphism Problem for groups where we are additionally given an almost $d$-ary sequence of invariant partitions.
Hence, in the remainder of this work, we shall be concerned with solving the latter problem.
The basic approach to do this is to adapt the Local Certificates Routine developed by Babai for his quasipolynomial time isomorphism test~\cite{Babai16}.

\section{Affected Orbits}\label{sec:affected:orbits}

The basis of Babai's Local Certificates algorithm is a group theoretic statement, the Unaffected Stabilizers Theorem (see \cite[Theorem 6]{Babai16}).
In the following we generalize this theorem to our setting. For the proof we roughly follow the argumentation from \cite{Babai15-full}.
However, on the technical level, several details need to be changed to allow for the treatment of the semi-regular operations allowed in our setting.

\begin{lemma}[cf.\;\cite{Babai15-full,Meierfrankenfeld95}]
 \label{la:epimorphism-subdirect-product}
 Let $G \leq K_1 \times \dots \times K_\ell$ be a subdirect product and let $\varphi\colon G \rightarrow S$ be an epimorphism where $S$ is a non-abelian simple group.
 Furthermore let $\pi_i\colon G \rightarrow K_i$ be the projection to the $i$-th component and $M_i = \ker(\pi_i)$.
 Then there is some $i^{*} \in [\ell]$ such that $M_{i^*} \leq \ker(\varphi)$.
\end{lemma}

\begin{lemma}
 \label{la:homomorphism-intersection}
 Let $G$ be a group, $H,K\unlhd G$ and suppose $\varphi \colon G \rightarrow S$ is an epimorphism where $S$ is a non-abelian simple group.
 Furthermore suppose that $H^{\varphi} = K^{\varphi} = S$.
 Then $(H \cap K)^{\varphi} = S$.
\end{lemma}

\begin{proof}
 Let $N = \ker(\varphi)$.
 Suppose that $(H \cap K)^{\varphi} \neq S$.
 Since $H \cap K \unlhd G$ and $S$ is a simple group we conclude that $(H \cap K)^{\varphi} = \{1\}$, that is, $H \cap K \leq N$.
 
 Now let $s_1,s_2 \in S$ be two arbitrary elements.
 Then there are $h \in H$, $k \in K$ such that $\varphi(h) = s_1$ and $\varphi(k) = s_2$.
 Moreover, $n:=h^{-1}k^{-1}hk \in H \cap K\le N$ since $H \unlhd G$ and $K \unlhd G$.
 Note that $hk = khn$.
 But then
 \[s_1s_2 = \varphi(h)\varphi(k) = \varphi(hk) = \varphi(khn) = \varphi(k)\varphi(h)\varphi(n) = s_2s_1.\]
 Since $s_1,s_2 \in S$ were chosen arbitrarily it follows that $S$ is abelian.
\end{proof}

\begin{lemma}[\cite{Babai15-full}, Lemma 8.3.1]
 \label{la:giant-transitive-babai}
 Let $G \leq S_d$ be a transitive group and $\varphi\colon G \rightarrow A_k$ an epimorphism where $k > \max\{8,2+\log_2 d\}$.
 Then $G_\alpha^{\varphi} \neq A_k$ for all $\alpha \in [d]$.
\end{lemma}

\begin{lemma}
 \label{la:epimorphism-transitive-almost-tree}
 Let $G \leq \Sym(\Omega)$ be a transitive group and suppose
 there is an almost $d$-ary sequence of invariant partitions $\{\Omega\} = \mathfrak{B}_0 \succ \dots \succ \mathfrak{B}_m = \{\{\alpha\} \mid \alpha \in \Omega\}$.
 Furthermore let $k > \max\{8,2 + \log_2d\}$, and let $\varphi\colon G \rightarrow A_k$ be an epimorphism.
 Then $G_\alpha^{\varphi} \neq A_k$ for all $\alpha \in \Omega$.
\end{lemma}

\begin{proof}
 We prove the statement by induction on the cardinality of $G$.
 Let $K = G_{(\mathfrak{B}_1)}$ be the normal subgroup stabilizing the block system $\mathfrak{B}_1$ and $N = \ker(\varphi)$.
 Observe that $N$ is a maximal normal subgroup of $G$
 ($N \unlhd G$ is a maximal normal subgroup of $G$ if and only if the quotient group $G/N$ is simple; here $G/N$ is isomorphic to $G^{\varphi} = A_k$).
 Hence, it holds that $K \leq N$ or $\langle K,N \rangle = KN = G$.
 
 First suppose $K \leq N$.
 Then $\varphi$ factors across $G \rightarrow G^{\mathfrak{B}_1} \stackrel{\psi}{\rightarrow} A_k$.
 Observe that $\psi$ is an epimorphism since $\varphi$ is an epimorphism.
 Suppose $|\mathfrak{B}_1| \leq d$.
 Then, by Lemma \ref{la:giant-transitive-babai}, for every $B \in \mathfrak{B}_1$ it holds that $(G^{\mathfrak{B}_1})_{B}^{\psi} \neq A_k$.
 Hence, $G_\alpha^{\varphi} \leq G_{B}^{\varphi} \neq A_k$ where $B \in \mathfrak{B}_1$ is the unique set such that $\alpha \in B$.
 Otherwise $G^{\mathfrak{B}_1}$ is semi-regular and hence, $(G^{\mathfrak{B}_1})_{B}^{\psi} = \{1\} \neq A_k$ for all $B \in \mathfrak{B}_1$.
 Again, $G_\alpha^{\varphi} \leq G_{B}^{\varphi} \neq A_k$ where $B \in \mathfrak{B}_1$ is the unique set such that $\alpha \in B$.
 
 So consider the case that $KN = G$, that is, $K^{\varphi} = A_k$.
 Suppose towards a contradiction that there is some $\alpha \in \Omega$ such that $G_\alpha^{\varphi} = A_k$.
 Pick $B \in \mathfrak{B}_1$ such that $\alpha \in B$.
 In particular, $G_B^{\varphi} = A_k$.
 \begin{claim}
  $G_{(B)}^{\varphi} \neq A_k$.
  \proof
  Assume towards a contradiction that $G_{(B)}^{\varphi} = A_k$.
  Then, by Lemma \ref{la:homomorphism-intersection}, $K_{(B)}^{\varphi} = (G_{(B)} \cap K)^{\varphi} = A_k$ since $G_{(B)}\unlhd G_B$, $K \unlhd G_B$ and $K^{\varphi} = A_k$.
  
  On the other hand, let $\Omega_1,\dots,\Omega_\ell$ be the orbits of $K$.
  Let $\pi_i\colon K \rightarrow \Sym(\Omega_i)$ be the restriction of $K$ to $\Omega_i$, $K_i = \im(\pi_i)$ and $M_i = \ker(\pi_i)$.
  By Lemma \ref{la:epimorphism-subdirect-product} there is some $i \in [\ell]$ such that $M_i \leq N$.
  Since $G$ acts transitively on the blocks $\{\Omega_1,\dots,\Omega_\ell\}$ the groups $M_i$, $i \in [\ell]$, are conjugate subgroups in $G$ and therefore $M_i \leq N$ for all $i \in [\ell]$.
  Pick $i^{*} \in [\ell]$ such that $\alpha \in \Omega_{i^{*}}$.
  Since $M_{i^{*}} \leq N$ the epimorphism $\varphi|_K\colon K \rightarrow A_k$ factors across $K_{i^{*}}$ as $K \stackrel{\pi_{i^{*}}}{\rightarrow} K_{i^{*}} \stackrel{\psi}{\rightarrow} A_k$.
  Hence, $K_{i^{*}}^{\psi} = A_k$.
  Moreover, $\mathfrak{B}_1[\Omega_{i^{*}}] \succ \dots \succ \mathfrak{B}_m[\Omega_{i^{*}}]$ is an almost $d$-ary sequence of partitions for $K_{i^{*}}$.
  By the induction hypothesis it follows that $(K_{i^{*}})_\alpha^{\psi} \neq A_k$ and thus, $K_\alpha^{\varphi} \neq A_k$.
  But this is a contradiction since $K_{(B)}^{\varphi} \leq K_\alpha^{\varphi}$.
  \uend
 \end{claim}
 Since $G_{(B)}^{\varphi} \trianglelefteq G_{B}^{\varphi}$ it follows $G_{(B)}^{\varphi} = \{1\}$.
 So $\varphi|_{G_B}$ factors across $G_B \rightarrow G_B^{B} \stackrel{\psi}{\rightarrow} A_k$.
 Moreover, $\varphi|_{G_\alpha}$ factors across $G_\alpha \rightarrow G_\alpha^{B} \stackrel{\psi'}{\rightarrow} A_k$, where~$\psi'= \psi|_{{G_\alpha^{B}}}$.
 Overall this means $(G_B^{B})^{\psi} = A_k$ and
 $(G_B^{B})_\alpha^{\psi} = (G_\alpha^{B})^{\psi'} = A_k$.
 But this contradicts the induction hypothesis since $\mathfrak{B}_1[B] \succ \dots \succ \mathfrak{B}_m[B]$ is an almost $d$-ary sequence of $G_B^{B}$-invariant partitions and $G_B^{B}$ is transitive.
\end{proof}

The following lemma shows that we can drop the assumption of $G$ being
transitive in Lemma~\ref{la:epimorphism-transitive-almost-tree} if we are only looking for some element $\alpha \in \Omega$ such that $G_\alpha^{\varphi} \neq A_k$.

\begin{lemma}
 \label{la:epimorphism-to-ak}
 Let $G \leq \Sym(\Omega)$ be a group and suppose there is an almost $d$-ary
 sequence of $G$-invariant partitions $\{\Omega\} = \mathfrak{B}_0 \succ \dots \succ \mathfrak{B}_m = \{\{\alpha\} \mid \alpha \in \Omega\}$.
 Furthermore let $k > \max\{8,2 + \log_2d\}$, and let $\varphi\colon G \rightarrow A_k$ be an epimorphism.
 Then $G_\alpha^{\varphi} \neq A_k$ for some $\alpha \in \Omega$.
\end{lemma}

\begin{proof}
 Let $\Omega_1,\dots,\Omega_\ell$ be the orbits of $G$ and let $\pi_i\colon G \rightarrow \Sym(\Omega_i)$ be the restriction of $G$ to $\Omega_i$.
 Let $G_i = \im(\pi_i)$ and $M_i = \ker(\pi_i)$.
 By Lemma \ref{la:epimorphism-subdirect-product} there is some $i \in [\ell]$ such that $M_i \leq \ker(\varphi)$.
 So $\varphi$ factors across $G_i$ as $G \stackrel{\pi_{i}}{\rightarrow} G_i \stackrel{\psi}{\rightarrow} A_k$.
 It follows that $G_i^{\psi} = A_k$.
 Now let $\alpha \in \Omega_i$.
 Note that $\mathfrak{B}_0[\Omega_i] \succeq \dots \succeq \mathfrak{B}_m[\Omega_i]$ forms an almost $d$-ary sequence of $G_i$-invariant partitions.
 Thus by Lemma \ref{la:epimorphism-transitive-almost-tree} it follows that $(G_i)_\alpha^{\psi} \neq A_k$ and thus, $G_\alpha^{\varphi} \neq A_k$.
\end{proof}

For a set $\Delta$ we denote by $\Alt(\Delta)$ the alternating group acting with its standard action on the set $\Delta$.
Moreover, we refer to the groups $\Alt(\Delta)$ and $\Sym(\Delta)$ as the \emph{giants} where $\Delta$ is an arbitrary finite set.

\begin{definition}[Babai \cite{Babai16}]
 Let $G \leq \Sym(\Omega)$.
 A homomorphism $\varphi\colon G \rightarrow S_k$ is a \emph{giant representation} if $G^\varphi \geq A_k$.
 In this case an element $\alpha \in \Omega$ is \emph{affected by $\varphi$} if $G_\alpha^{\varphi} \not\geq A_k$.
\end{definition}

\begin{remark}
 \label{rem:affected-orbit}
 Let $\varphi\colon G \rightarrow S_k$ be a giant representation and suppose $\alpha \in \Omega$ is affected by $\varphi$.
 Then every element in the orbit $\alpha^{G}$ is affected by
 $\varphi$. We call $\alpha^G$ an \emph{affected orbit} (with respect to $\varphi$).
\end{remark}

With this definition we can now state the generalization of the Unaffected Stabilizers Theorem (see \cite[Theorem 6]{Babai16}).

\begin{theorem}
 \label{thm:unaffected-stabilizer-tree}
 Let $G \leq \Sym(\Omega)$ be a permutation group and suppose there is an almost $d$-ary sequence of $G$-invariant partitions
 $\{\Omega\} = \mathfrak{B}_0 \succ \dots \succ \mathfrak{B}_m = \{\{\alpha\} \mid \alpha \in \Omega\}$.
 Furthermore let $k > \max\{8,2 + \log_2d\}$ and $\varphi\colon G \rightarrow S_k$ be a giant representation.
 Let $D \subseteq \Omega$ be the set of elements not affected by $\varphi$.
 Then $G_{(D)}^{\varphi} \geq A_k$.
\end{theorem}

\begin{proof}
 First suppose that $G^{\varphi} = A_k$.
 The set $D$ is $G$-invariant (cf.\ Remark \ref{rem:affected-orbit}).
 Let $\psi\colon G \rightarrow \Sym(D)$ be the restriction of $G$ to $D$. Observe that $\ker(\psi) = G_{(D)}$.
 So $G_{(D)} \trianglelefteq G$ and hence, $G_{(D)}^{\varphi} \trianglelefteq G^{\varphi} = A_k$.
 Assume towards a contradiction that $G_{(D)}^{\varphi} \neq A_k$.
 Then $G_{(D)}^{\varphi} = 1$, that is, $G_{(D)} \leq \ker(\varphi)$.
 So $\varphi$ factors across $H := G^{\psi} \leq \Sym(D)$ as $G \stackrel{\psi}{\rightarrow} H \stackrel{\rho}{\rightarrow} A_k$.
 Note that $\mathfrak{B}_0[D] \succeq \dots \succeq \mathfrak{B}_m[D]$ forms an almost $d$-ary sequence of $H$-invariant partitions.
 It follows that $H^{\rho} = A_k$ and hence, $H^{\rho}_\alpha \neq A_k$ for some $\alpha \in D$ by Lemma \ref{la:epimorphism-to-ak}.
 But $G_\alpha^{\varphi} = H_\alpha^{\rho} = A_k$ since $\alpha \in D$ is not affected, which is a contradiction.
 
 So consider the case that $G^{\varphi} = S_k$ and let $G' = \varphi^{-1}(A_k)$. 
 Let $\varphi' = \varphi|_{G'}$. 
 Let $D'$ be the set of points not affected by $\varphi'$.
 We argue that $D' = D$.
 We have $D' \subseteq D$ because $G_\alpha^{\varphi} \geq (G_\alpha')^{\varphi}$ for all $\alpha \in \Omega$.
 Now suppose there is some $\alpha \in D \setminus D'$.
 Then $G_\alpha^{\varphi} \geq A_k$, $(G_\alpha')^{\varphi} < A_k$ and
 $|G_\alpha^{\varphi} : (G_\alpha')^{\varphi}| \leq 2$.
 Overall, this gives us a subgroup of $A_k$ of index $2$. But such a subgroup would be a normal subgroup contradicting the fact that $A_k$ is simple.
 So $D' = D$.
 Then, by the previous case, $G_{(D)}^{\varphi} \geq (G')_{(D)}^{\varphi'} = (G')_{(D')}^{\varphi'} \geq A_k$.
\end{proof}

We also use Babai's Affected Orbit Lemma, which does not need to be adapted to our setting.

\begin{theorem}[\mbox{\cite[Theorem 6(b)]{Babai16}}]
 \label{thm:kernel-affected-orbits}
 Let $G \leq \Sym(\Omega)$ be a permutation group and suppose $\varphi\colon G \rightarrow S_k$ is a giant representation for $k \geq 5$.
 Suppose $\Delta \subseteq \Omega$ is an affected orbit of $G$ (with respect to $\varphi$). Then every orbit of $\ker(\varphi)$ in $\Delta$ has length at most $|\Delta|/k$.
\end{theorem}

\section{Local Certificates}
\label{sec:local-certificates}

In this section we adapt the Local Certificates technique developed in \cite{Babai16} to our setting using the generalization to the Unaffected Stabilizers Theorem presented in the previous section (Theorem~\ref{thm:unaffected-stabilizer-tree}).
As before the basic argumentation and also the notation follows \cite{Babai15-full}.
Besides an adaptation to our setting, the main difference is a more precise analysis of the running time which is required for our overall analysis.

\subsection{Algorithm}

Let $G \leq \Sym(\Omega)$ be a permutation group and let $\mathfrak{x} \colon \Omega \rightarrow \Sigma$ be a string.
Furthermore let $\varphi\colon G \rightarrow S_k$ be a giant representation.
For a set $T \subseteq [k]$ let $G_T = \varphi^{-1}((G^{\varphi})_T)$ and $G_{(T)} = \varphi^{-1}((G^{\varphi})_{(T)})$.

Since our notation may be getting hard to trace here, as an example, let us disassemble it
for one of the groups appearing in the next definition:
$((\Aut_{G_T}(\mathfrak{x}))^{\varphi})^{T}$. We start from the group
$G\le\Sym(\Omega)$. With the homomorphism $\varphi$ we
map it to $G^\varphi\le S_k$ acting on the set $[k]$. Then we take the setwise stabilizer
of the set $T\subseteq[k]$ and obtain the subgroup $G^\varphi_T\le S_k$. We pull back to $\Omega$ via $\varphi^{-1}$ and obtain the
subgroup $G_T:=\varphi^{-1}(G^\varphi_T)\le \Sym(\Omega)$.
We move on to the subgroup $\Aut_{G_T}(\mathfrak x)$ and, once more,
transfer it back to $[k]$ via $\varphi$ to obtain $(\Aut_{G_T}(\mathfrak{x}))^{\varphi}\le S_k$.
The set $T$ is invariant with respect to this group, so the group also acts on $T$.
This, finally, gives us $((\Aut_{G_T}(\mathfrak{x}))^{\varphi})^{T}\le\Sym(T)$.
What makes this complicated is the going back and forth between $\Omega$ and $[k]$.
But this interplay between the two sets, or rather the groups acting on these sets, is crucial for the overall reasoning.

\begin{definition}\label{def:full}
 A set $T \subseteq [k]$ is \emph{full} if $((\Aut_{G_T}(\mathfrak{x}))^{\varphi})^{T} \geq \Alt(T)$.
 A \emph{certificate of fullness} is a subgroup $K \leq \Aut_{G_T}(\mathfrak{x})$ such that $(K^{\varphi})^{T} \geq \Alt(T)$.
 A \emph{certificate of non-fullness} is a non-giant $M \leq \Sym(T)$ such that $((\Aut_{G_T}(\mathfrak{x}))^{\varphi})^{T} \leq M$.
\end{definition}

Let $W \subseteq \Omega$ be $G$-invariant and let $\mathfrak{y}\colon \Omega \rightarrow \Sigma$ be a second string.
Recall that $\Iso_G^{W}(\mathfrak{x},\mathfrak{y}) = \{g \in G \mid
\forall \alpha \in W\colon\mathfrak{x}(\alpha) =
\mathfrak{y}(\alpha^{g})\}$ and $\Aut_G^{W}(\mathfrak{x}) = \Iso_G^{W}(\mathfrak{x},\mathfrak{x})$.

For $H \leq G$ we define $\Aff(H,\varphi) = \{\alpha \in \Omega \mid H_\alpha^{\varphi} \not\geq A_k\}$.
Note that for $H_1 \leq H_2 \leq G$ it holds that $\Aff(H_1,\varphi) \supseteq \Aff(H_2,\varphi)$.

Finally, recall that $n$ always denotes the size of the permutation domain $\Omega$.

\begin{lemma}
 \label{la:local-certificate}
 Let $\mathfrak{x}\colon \Omega \rightarrow \Sigma$ be a string, $G \leq \Sym(\Omega)$ be a group and suppose there is
 an almost $d$-ary sequence of $G$-invariant partitions $\{\Omega\} = \mathfrak{B}_0 \succ \dots \succ \mathfrak{B}_m = \{\{\alpha\} \mid \alpha \in \Omega\}$.
 Furthermore suppose there is a giant representation $\varphi\colon G \rightarrow S_k$ and let $T \subseteq [k]$ be a set of size $|T| = t > \max\{8,2 + \log_2d\}$.
 
 Then there are natural numbers $n_1,\dots,n_\ell \leq n/2$ such that $\sum_{i=1}^{\ell}n_i \leq n$ and,
 for each $i \in [\ell]$ using at most $t!$ recursive calls to String
 Isomorphism over domain size $n_i$ and $\mathcal{O}(t!\cdot n^{c})$
 additional computation, one can decide whether $T$ is full or not and generate a corresponding certificate.
\end{lemma}

\begin{algorithm}[t]
 \caption{\textsf{LocalCertificates}}
 \label{alg:local-certificates}
 \DontPrintSemicolon
 \SetKwInOut{Input}{Input}
 \SetKwInOut{Output}{Output}
 \Input{$G \leq \Sym(\Omega)$, $\mathfrak{x}\colon \Omega \rightarrow \Sigma$, and $\varphi\colon G \rightarrow S_k$ with $k > \max\{8,2 + \log_2d\}$.
        There exists an almost $d$-ary sequence of $G$-invariant partitions $\{\Omega\} = \mathfrak{B}_0 \succ \dots \succ \mathfrak{B}_m = \{\{\alpha\} \mid \alpha \in \Omega\}$.
        }
 \Output{non-giant $M \leq S_k$ with $(\Aut_G(\mathfrak{x}))^{\varphi} \leq M$ or $K \leq \Aut_G(\mathfrak{x})$ with $K^{\varphi} \geq A_k$.}
 \BlankLine
 $G_0 := G$\;
 $W_0 := \emptyset$\;
 $i := 0$\;
 \While{$G_i^{\varphi} \geq A_k \textup{ \bfseries and } W_{i} \neq \Aff(G_i,\varphi)$}{
  $W_{i+1} := \Aff(G_i, \varphi)$\;
  $W_{i+1}^{*} := W_{i+1} \setminus W_i$\;
  \eIf{$|W_{i+1}^{*}| \leq \frac{1}{2}|\Omega|$}{
   $G_{i+1} := \Aut_{G_i}^{W_{i+1}^{*}}(\mathfrak{x})$\;
  }{
   $G_{i+1} := \emptyset$\;
   $N := \ker(\varphi|_{G_i})$\;
   \For{$g \in G_i^{\varphi}$}{
    compute a $\bar g \in \varphi^{-1}(g)$\;
    $G_{i+1} := G_{i+1} \cup \Aut_{N\bar g}^{W_{i+1}^{*}}(\mathfrak{x})$
   }
  }
  $i := i+1$\;
 }
 \eIf{$G_i^{\varphi} \not\geq A_k$}{
  \Return $G_i^{\varphi}$\;
 }{
  \Return $(G_i)_{(\Omega \setminus W_i)}$
 }
\end{algorithm}

\begin{proof}
 Without loss of generality assume $T = [k]$.
 Otherwise, we can compute the group $G_T$ and restrict the image of $\varphi$ to the set $T$.
 
 Consider Algorithm \ref{alg:local-certificates}.
 The algorithm computes, for increasing windows $W_0 \subseteq W_1 \subseteq W_2 \subseteq \dots$,
 the group $G_i$ of permutations that respect the input string $\mathfrak{x}$ on the window $W_i$,
 that is, $G_i = \Aut_{G}^{W_{i}}(\mathfrak{x}) = \Aut_{G_{i-1}}^{W_{i}^{*}}(\mathfrak{x})$, where $W_i^*=W_i\setminus W_{i-1}$.
 Note that $G_i \leq G_{i-1}$ and $W_{i+1} \supseteq W_i$ for $i \geq 1$
 (initially $W_1 \neq \emptyset$ since by Lemma~\ref{la:epimorphism-to-ak} at least one point has to be affected). 
 The algorithm stops when the current group $G_i^{\varphi}$ is not a giant or the window stops growing.
 
 Let $i^{*}$ be the value of the variable $i$ at the end of while-loop.
 Furthermore let $W = W_{i^{*}}$.
 Note that $\{W_j^{*} \mid 1 \leq j \leq i^{*}\}$ forms a partition of the set $W$.
 
 We first show the correctness of the algorithm.
 For every $0 \leq j \leq i^{*}$ it holds that $\Aut_G(\mathfrak{x}) \leq G_j \leq G$.
 We distinguish two cases.
 First suppose that $G_{i^{*}}^{\varphi} \not\geq A_k$.
 Then $G_{i^{*}}^{\varphi}$ forms a certificate of non-fullness.
 Otherwise $G_{i^{*}}^{\varphi} \geq A_k$ and $W = \Aff(G_{i^{*}},\varphi)$.
 Note that $\mathfrak{B}_0,\dots,\mathfrak{B}_m$ forms an almost $d$-ary sequence of invariant partitions for the group $G_{i^{*}}$ (cf.\ Observation \ref{obs:sequence-of-partitions}).
 So $((G_{i^{*}})_{(\Omega \setminus W)})^{\varphi} \geq A_k$ by Theorem~\ref{thm:unaffected-stabilizer-tree}.
 Furthermore, it easy to check that $G_{i^*}$ respects the string $\mathfrak{x}$ on all positions in $W_j$ for all $0 \leq j \leq i^{*}$.
 Hence, $(G_{i^{*}})_{(\Omega \setminus W)} \leq \Aut_G(\mathfrak{x})$ because it respects all positions within $W$ and fixes all other positions.
 
 It remains to analyze the running time of the algorithm. 
 Again we distinguish two cases.
 First suppose $|W_j^{*}| \leq n/2$ for all $j \in [i^{*}]$.
 Then, for each $j \in [i^{*}]$, the algorithm makes one recursive call to String Isomorphism over domain size $|W_i^{*}| \leq n/2$ (Line 8) and $\sum_{j \in [i^{*}]} |W_j^{*}| \leq |W| \leq n$.
 Otherwise there is a unique $j^{*} \in \{0,\dots,i^{*}-1\}$ such that $|W_{j^{*}+1}^{*}| > n/2$.
 Let $N = \ker(\varphi|_{G_{j^{*}}})$.
 Since all elements in $W_{j^{*}+1}^{*}$ are affected by $\varphi$ with respect to $G_{j^{*}}$ it holds that
 every orbit of $N$ in $W_{j^{*}+1}^{*}$ has size at most $|W_{j^{*}+1}^{*}|/k$ by Theorem \ref{thm:kernel-affected-orbits}.
 Since $G_{j^{*}}^\varphi\le S_k$, for each orbit the algorithm makes at most $k!$ calls to String Isomorphism
 where the domain is restricted to exactly this orbit (Line 14).
 (Recall the assumption $T=[k]$, which implies $t!=k!$.)
 Additionally, for every $j \in [i^{*}], j \neq j^{*}+1$ there is one recursive call to String Isomorphism over domain size $|W_j^{*}|$.
\end{proof}

\subsection{Comparing Local Certificates}

\begin{lemma}
 \label{la:compare-local-certificates}
 Let $\mathfrak{x}_1,\mathfrak{x}_2\colon \Omega \rightarrow \Sigma$ be two strings,
 $G \leq \Sym(\Omega)$ be a group and suppose there is an almost $d$-ary sequence of $G$-invariant partitions
 $\{\Omega\} = \mathfrak{B}_0 \succ \dots \succ \mathfrak{B}_m = \{\{\alpha\} \mid \alpha \in \Omega\}$.
 Furthermore suppose there is a giant representation $\varphi\colon G \rightarrow S_k$.
 Let $T_1,T_2 \subseteq [k]$ be sets of equal size $t = |T_1| = |T_2| > \max\{8,2 + \log_2d\}$ and suppose $T_1$ is not full.
 
 Then there are natural numbers $n_1,\dots,n_\ell \leq n/2$ such that $\sum_{i=1}^{\ell}n_i \leq n$ and,
 for each $i \in [\ell]$ using $t!$ recursive calls to String
 Isomorphism over domain size $n_i$ and $\mathcal{O}(t!n^{c})$
 additional computation,
 one can compute a non-giant group $M \leq \Sym(T_1)$ and a bijection $\sigma\colon T_1 \rightarrow T_2$ such that
 \begin{equation}\label{eq:clc1}
  \left\{g^\varphi|_{T_1} \,\middle|\, g \in \Iso_G(\mathfrak{x}_1,\mathfrak{x}_2) \wedge T_1^{\left(g^{\varphi}\right)} = T_2\right\} \subseteq M\sigma.
 \end{equation}
 Moreover, the set of bijections $M\sigma$ is canonical for the two test sets
 (w.r.t.~$\mathfrak{x}_1,\mathfrak{x}_2,G$ and the giant representation $\varphi$).
\end{lemma}

Here, canonical means that given additional test sets $T_1',T_2' \subseteq [k]$ such that $T_i' = T_i^{g}$ for some $g \in \Aut_G(\mathfrak{x_i})$ for both $i \in \{1,2\}$, the algorithm computes a set $M'\sigma'$ such that $(M\sigma)^{g} = M'\sigma'$.

\begin{algorithm}
 \caption{\textsf{CompareLocalCertificates}}
 \label{alg:compare-local-certificates}
 \DontPrintSemicolon
 \SetKwInOut{Input}{Input}
 \SetKwInOut{Output}{Output}
 \Input{$G \leq \Sym(\Omega)$, $\mathfrak{x}_1,\mathfrak{x}_2\colon \Omega \rightarrow \Sigma$,
        $\varphi\colon G \rightarrow S_k$, and $T_1,T_2\subseteq [k]$ of size $t > \max\{8,2 + \log_2d\}$.
        There exists an almost $d$-ary sequence of invariant partitions
                $\{\Omega\} = \mathfrak{B}_0 \succ \dots \succ \mathfrak{B}_m = \{\{\alpha\} \mid \alpha \in \Omega\}$ and~$T_1$ is not full.
        }
 \Output{non-giant group $M \leq \Sym(T_1)$ and bijection $\sigma\colon T_1 \rightarrow T_2$ such that
         \[\left\{g^\varphi|_{T_1} \,\middle|\, g \in \Iso_G(\mathfrak{x}_1,\mathfrak{x}_2) \wedge T_1^{\left(g^{\varphi}\right)} = T_2\right\} \subseteq M\sigma.\]}
 \BlankLine
 compute $\sigma_0 \in G$ such that $T_1^{(\sigma_0^{\varphi})} = T_2$\;
 $G_0 := G_{T_1}$\;
 $W_0 := \emptyset$\;
 $i := 0$\;
 $\psi \colon G_0 \rightarrow \Sym(T_1)$ homomorphism obtained from $\varphi$ by restricting the image to $T_1$\;
 \While{$G_i^{\psi} \geq \Alt(T_1)$}{
  $W_{i+1} := \Aff(G_i, \psi)$\;
  $W_{i+1}^{*} := W_{i+1} \setminus W_i$\;
  \eIf{$|W_{i+1}^{*}| \leq \frac{1}{2}|\Omega|$}{
   $G_{i+1}\sigma_{i+1} := \Iso_{G_i\sigma_i}^{W_{i+1}^{*}}(\mathfrak{x},\mathfrak{y})$\;
  }{
   $G_{i+1} := \emptyset$\;
   $N := \ker(\psi|_{G_i})$\;
   $\ell := 0$\;
   \For{$g \in G_i^{\psi}$}{
    compute $\bar g \in \psi^{-1}(g)$\;
    $H_\ell h_\ell := \Iso_{N\bar g\sigma_i}^{W_{i+1}^{*}}(\mathfrak{x},\mathfrak{y})$\;
    $\ell := \ell+1$\;
   }
   $G_{i+1}\sigma_{i+1} := \bigcup_{j \leq \ell} H_jh_j$\;
  }
  $i := i+1$\;
 }
 
 \Return $(G_i^{\psi},(\sigma_i^{\varphi})|_{T_1})$\;
\end{algorithm}

\begin{proof}
 Consider Algorithm \ref{alg:compare-local-certificates}.
 First suppose towards a contradiction there is some $i$ such that $W_{i+1} = W_i$.
 Then $((G_i)_{(\Omega \setminus W_i)})^{\psi} \geq \Alt(T_1)$ by Theorem \ref{thm:unaffected-stabilizer-tree}.
 Furthermore $(G_i)_{(\Omega \setminus W_i)} \leq \Aut_G(\mathfrak{x})$.
 Together this implies that $(\Aut_{G_{T_1}}(\mathfrak{x}))^{\psi} \geq \Alt(T_1)$ contradicting the fact that $T_1$ is not full.
 
 So the algorithm terminates and returns a non-giant group $M \leq \Sym(T_1)$ and a bijection $\sigma\colon T_1 \rightarrow T_2$ with the desired properties.
 The complexity analysis is completely analogous to Lemma \ref{la:local-certificate}.
 
 Finally, the canonicity of the set of bijections $M\sigma$ follows from the fact that in each iteration the set of affected points is canonically defined.
\end{proof}

\subsection{Aggregating Local Certificates}

Let $G \leq \Sym(\Omega)$ be a group.
The \emph{symmetry defect} of $G$ is the minimal $t \in [n]$ such that there is a set $\Delta \subseteq \Omega$ of size $|\Delta| = n-t$ such that $\Alt(\Delta) \leq G$ (the group $\Alt(\Delta)$ fixes all elements of $\Omega\setminus \Delta$).
In this case the \emph{relative symmetry defect} is $t/n$.

For any relational structure $\mathfrak{A}$ we define the \emph{(relative) symmetry defect} of $\mathfrak{A}$ to be the (relative) symmetry defect of its automorphism group $\Aut(\mathfrak{A})$.

\begin{theorem}[cf.\;\cite{DM96}, Theorem 5.2 A,B]
 Let $A_n \leq S \leq S_n$ and suppose $n > 9$.
 Let $G \leq S$ and $r < n/2$.
 Suppose that $|S : G| < \binom{n}{r}$.
 Then the symmetry defect of $G$ is strictly less than $r$.
\end{theorem}

Using the inequality $\binom{n}{\lfloor n/4\rfloor} \geq \left(\frac{n}{\lfloor n/4\rfloor}\right)^{\lfloor n/4\rfloor} \geq \frac{1}{4} \cdot \sqrt{2}^{n}$ we get the following corollary.

\begin{corollary}
 \label{cor:index-subgroup-large-symmetry-defect}
 Let $A_n \leq S \leq S_n$ and suppose $n \geq 24$.
 Let $G \leq S$ and suppose the relative symmetry defect of $G$ is at least $1/4$.
 Then $|S:G| \geq (4/3)^{n}$.
\end{corollary}

\begin{lemma}
 \label{la:aggregate-local-certificates}
 Let $\mathfrak{x}_1,\mathfrak{x}_2\colon \Omega \rightarrow \Sigma$ be two strings,
 $G \leq \Sym(\Omega)$ be a group and suppose there is an almost $d$-ary sequence of $G$-invariant partitions
 $\{\Omega\} = \mathfrak{B}_0 \succ \dots \succ \mathfrak{B}_m = \{\{\alpha\} \mid \alpha \in \Omega\}$.
 Furthermore suppose there is a giant representation $\varphi\colon G \rightarrow S_k$.
 Let $\max\{8,2 + \log_2d\} < t < k/10$.
 
 Then there are natural numbers $\ell \in \mathbb{N}$ and $n_1,\dots,n_\ell \leq n/2$ such that $\sum_{i=1}^{\ell}n_i \leq k^{\mathcal{O}(t)}n$ and,
 for each $i \in [\ell]$ using a recursive call to String Isomorphism
 over domain size $n_i$, and $k^{\mathcal{O}(t)}n^{c}$
 additional computation,
 one obtains for $i=1,2$ one of the following:
 \begin{enumerate}
  \item\label{item:aggregate-local-certificates-1} a family of $r\le k^6$ many $t$-ary relational structures $\mathfrak{A}_{i,j}$, for $j\in[r]$, associated with $\mathfrak x_i$,
    each with domain $D_{i,j}\subseteq[k]$ of size $|D_{i,j}| \geq \frac{3}{4}k$ and with relative symmetry defect at least $\frac{1}{4}$ such that
    \[\left\{\mathfrak{A}_{1,1},\dots,\mathfrak{A}_{1,r}\right\}^{\varphi(g)} = \left\{\mathfrak{A}_{2,1},\dots,\mathfrak{A}_{2,r}\right\} \text{ for every } g \in \Iso_G(\mathfrak{x}_1,\mathfrak{x}_2),\]
    or
  \item\label{item:aggregate-local-certificates-2} a subset
    $\Delta_i \subseteq [k]$ associated with $\mathfrak{x}_i$ of size
    $|\Delta_i| \geq \frac{3}{4}k$ and $K_i \leq \Aut_{G_{\Delta_i}}(\mathfrak{x}_i)$
    such that $(K_i^{\varphi})^{\Delta_i} \geq \Alt(\Delta_i)$ and
    \[\Delta_1^{\varphi(g)} = \Delta_2 \text{ for every } g \in \Iso_G(\mathfrak{x}_1,\mathfrak{x}_2).\]
 \end{enumerate}
\end{lemma}

The proof is completely analogous to the proof of \cite[Theorem
24]{Babai16} replacing the methods to compute the local certificates.
Note that colorings and equipartitions of a subset of $[k]$ can be viewed as relational structures.
For the sake of completeness a full proof of the lemma is given in Appendix \ref{app:proof-aggregate}.

Next we describe how we use the two possible outcomes of the previous lemma to make progress.

\begin{lemma}
 \label{la:find-structure}
 Suppose Option \ref{item:aggregate-local-certificates-1} of Lemma
 \ref{la:aggregate-local-certificates} is satisfied, yielding a number
 $r\le k^6$ and relational structures $\mathfrak A_{i,j}$ for
 $i\in[2],j\in[r]$.
 Then there are subgroups $H_j \leq G$ and elements $h_j \in \Sym(\Omega)$ for $j \in [r]$ such that
 \begin{enumerate}
  \item $|G^{\varphi} : H_j^{\varphi}| \geq (4/3)^{k}$ for all $j \in [r]$, and
  \item $\mathfrak{x}_1 \cong_G \mathfrak{x}_2$ if and only if $\mathfrak{x}_1 \cong_{H_jh_j} \mathfrak{x}_{2}$ for some $j \in [r]$, 
        and given representations for the sets $\Iso_{H_jh_j}(\mathfrak{x}_1,\mathfrak{x}_{2})$ for all $j \in [r]$ one can compute in polynomial time a representation for $\Iso_G(\mathfrak{x}_1,\mathfrak{x}_{2})$.
 \end{enumerate}
 Moreover, given the relational structures $\mathfrak{A}_{i,j}$ for all $i \in [2]$ and $j \in [r]$, the groups $H_j$ and elements $h_j$ can be computed in time $k^{\mathcal{O}(t^{c} (\log k)^{c})}n^{c}$ for some constant $c$.
\end{lemma}

\begin{proof}
 Let $D_{i,j} \subseteq [k]$ be the domain of $\mathfrak{A}_{i,j}$ for all $i \in [2]$ and $j \in [r]$.
 Let $\mathfrak{A}_1 = \mathfrak{A}_{1,1}$ and $D_1 = D_{1,1}$.
 Now define \[H_jh_j = \{g \in G \mid (D_1)^{\left(g^{\varphi}\right)} = D_{2,j} \wedge (g^{\varphi})|_{D_1} \in \Iso(\mathfrak{A}_1,\mathfrak{A}_{2,j})\}.\]
 Using the quasipolynomial time isomorphism test from \cite{Babai16} the set $\Iso(\mathfrak{A}_1,\mathfrak{A}_{2,j})$ can be computed in time $k^{\mathcal{O}(t^{c} (\log k)^{c})}$ for some constant $c$
 (first translate the relational structures into two graphs of size $k^{\mathcal{O}(t)}$ (see e.g.\ \cite{Miller79}) and then apply the isomorphism test from \cite{Babai16} to the resulting graphs).
 Hence, the groups $H_j \leq G$ and elements $h_j \in \Sym(\Omega)$ can be computed within the desired time bound.
 Moreover
 \[\Iso_G(\mathfrak{x}_1,\mathfrak{x}_2) = \bigcup_{j \in [r]} \Iso_{H_jh_j}(\mathfrak{x}_1,\mathfrak{x}_2).\]
 Finally observe that the symmetry defect of $H_j^{\varphi}$ is at least $\frac{1}{4}$.
 So $|G^{\varphi} : H_j^{\varphi}| \geq (4/3)^{k}$ by Corollary \ref{cor:index-subgroup-large-symmetry-defect}.
\end{proof}

\begin{remark}
 The proof of the previous lemma is the only place where we use Babai's quasipolynomial time isomorphism test \cite{Babai16} as a black box.
\end{remark}

\begin{lemma}
 \label{la:find-symmetry}
 Suppose Option \ref{item:aggregate-local-certificates-2} of Lemma \ref{la:aggregate-local-certificates} is satisfied.
 Then there is a number $r \in \{1,2\}$, a subgroup $H \leq G$ and elements $h_j \in \Sym(\Omega)$ for $j \in [r]$ such that
 \begin{enumerate}
  \item $|G^{\varphi} : H^{\varphi}| \geq (4/3)^{k}$, and
  \item $\mathfrak{x}_1 \cong_G \mathfrak{x}_2$ if and only if
    $\mathfrak{x}_1 \cong_{Hh_j} \mathfrak{x}_{2}$  for some $j \in [r]$, 
    and given representations for the sets
    $\Iso_{Hh_j}(\mathfrak{x}_1,\mathfrak{x}_{2})$ for all $j \in [r]$
    and a generating set for $K_1$
    one can compute in polynomial time a representation for $\Iso_G(\mathfrak{x}_1,\mathfrak{x}_{2})$.
 \end{enumerate}
 Moreover, given the sets $\Delta_{i}$ for all $i \in [2]$, the group $H$ and the elements $h_i$ can be computed in polynomial time.
\end{lemma}

\begin{proof}
 Let $H = G_{(\Delta_1)}$ (recall that $G_{(T)} = \varphi^{-1}((G^{\varphi})_{(T)})$ for $T \subseteq [k]$).
 Let $g \in G$ such that $\Delta_1^{g^{\varphi}} = \Delta_2$ and $\tau \in G_{\Delta_1}$ such that $(\tau^{\varphi})^{\Delta_1}$ is a transposition.
 Now define $h_1 = g$ and $h_2 = \tau g$.
 Then $\mathfrak{x}_1 \cong_G \mathfrak{x}_2$ if and only if
 $\mathfrak{x}_1 \cong_{Hh_j} \mathfrak{x}_{2}$ since
 $(K_1^{\varphi})^{\Delta_1} \geq \Alt(\Delta_1)$.
 Moreover, if $G_jg_j = \Iso_{Hh_j}(\mathfrak{x}_1,\mathfrak{x}_{2})$ then $\Iso_{G}(\mathfrak{x}_1,\mathfrak{x}_{2}) = \bigcup_{j=1,2}\langle K_1,G_j \rangle g_j$.
 Finally, $|G^{\varphi} : H^{\varphi}| \geq |\Alt(\Delta_1)| \geq (4/3)^{k}$.
\end{proof}

\section{String Isomorphism}
\label{sec:algorithm}

We are now ready to formalize our algorithm. We shall need the
following result characterizing the obstacle cases for efficient Luks reduction.

\begin{lemma}[cf.\;\cite{Babai15-full}, Theorem 3.2.1]
 \label{la:compute-giant-representation}
 Let $G \leq S_d$ be a primitive group of order $|G| \geq d^{1+\log d}$ 
 where $d$ is greater than some absolute constant.
 Then there is a polynomial-time algorithm computing a normal subgroup $N \leq G$ of index $|G:N| \leq d$, an $N$-invariant equipartition $\mathfrak{B}$
 and a giant representation $\varphi\colon N \rightarrow S_k$ where $k \geq \log d$ and $\ker(\varphi) = N_{(\mathfrak{B})}$.
\end{lemma}

\begin{lemma}
 \label{lem:partition-small-size-recursion}
 Let $G \leq \Sym(\Omega)$ be transitive and let $\mathfrak{x},\mathfrak{y} \colon \Omega \rightarrow \Sigma$ be two strings.
 Moreover, suppose there is an almost $d$-ary sequence of $G$-invariant partitions
 $\{\Omega\} = \mathfrak{B}_0 \succ \dots \succ \mathfrak{B}_m = \{\{\alpha\} \mid \alpha \in \Omega\}$
 such that $|\mathfrak{B}_1| \leq d$.
 Then there are natural numbers $\ell \in \mathbb{N}$ and $n_1,\dots,n_\ell \leq n/2$ such that $\sum_{i=1}^{\ell}n_i \leq 2^{\mathcal{O}((\log d)^{3})}n$ and,
 using a recursive call to String Isomorphism over domain size at most $n_i$ for each $i \in [\ell]$ and $d^{\mathcal{O}((\log d)^{c})}n^{c}$ additional computation,
 one can compute $\Iso_G(\mathfrak{x},\mathfrak{y})$.
\end{lemma}

\begin{proof}
 Let $\mathfrak{B} = \mathfrak{B}_1$ and let $P = G^{\mathfrak{B}}$ be
 the induced action of $G$ on the partition $\mathfrak{B}$.
 Without loss of generality suppose $P$ is primitive (otherwise replace $\mathfrak{B}$ with a block system of smaller size).
 First suppose $|P| \leq d^{1+\log d}$.
 Then the statement immediately follows by standard Luks reduction.
 Otherwise let $N \leq P$ be the normal subgroup computed by Lemma
 \ref{la:compute-giant-representation} and let $\mathfrak{C}$ be the corresponding partition and $\psi\colon N \rightarrow S_k$ the giant representation.
 Observe that $k \leq d$ since $N \leq P \leq \Sym(\mathfrak{B})$ and $|\mathfrak{B}| \leq d$.
 Let $G' = \{g \in G \mid g^{\mathfrak{B}} \in N\}$.
 Also let $\mathfrak{C}' = \{\{\alpha\in\Omega \mid \exists B \in C\colon
 \alpha \in B\} \mid C \in \mathfrak{C}\}$.
 Note that $\mathfrak{C}'$ is $G'$-invariant.
 Since $|G:G'| \leq d$ it suffices to prove the statement for the group $G'$.
 Let $\varphi\colon G' \rightarrow S_k\colon g \mapsto (g^{\mathfrak{B}})^{\psi}$.
 Note that $\varphi$ is a giant representation and $(G')_{(\mathfrak{C}')} = \ker(\varphi)$.
 Let $t = \max\{9,3 + \log d\}$.
 In case $k \leq 10t$ the statement follows again by standard Luks
 reduction (in this case $|G' : (G')_{(\mathfrak{C}')}| = |G' : \ker(\varphi)| \leq k! \leq 2^{\mathcal{O}((\log d)^{2})}$).
 So suppose $\max\{8,2 + \log d\} < t < k/10$.
 In this case the requirements of Lemma \ref{la:aggregate-local-certificates} are satisfied.
 
 Using Lemma \ref{la:aggregate-local-certificates},
 \ref{la:find-structure} and \ref{la:find-symmetry} we can reduce the
 problem (using additional recursive calls to String Isomorphism over
 domain size at most $n/2$) to at most $k^{6}$ instances of
 $H$-isomorphism over the same strings $\mathfrak x,\mathfrak y$ for
 groups $H\le G'$ with $|(G')^{\varphi} : H^{\varphi}| \geq (4/3)^{k}$.
 Applying the same argument to these instances of
 $H$-isomorphism and repeating the process until we can afford to perform standard Luks reduction gives our desired algorithm.
 It remains to analyze its running time, that is, we have to analyze the number of times this process has to be repeated until we reach a group that is sufficiently small to perform standard Luks reduction.
 Towards this end, we analyze the parameter $k$ of the giant representation and show that it has to be reduced in each round by a certain amount.
 Recall that our algorithm performs standard Luks reduction as soon as $k \leq 10t$.
 
 Consider the recursion tree of the algorithm (ignoring the additional
 recursive calls to String Isomorphism over domain size at most $n/2$ for the moment).
 Recall that $\mathfrak{C}'$ is $G'$-invariant and thus, it is also $H$-invariant.
 In case $H$ is not transitive it is processed orbit by orbit.
 Note that there is at most one orbit of size greater than $n/2$ that has to be considered in the current recursion (for the other orbits additional recursive calls to String Isomorphism over domain size at most $n/2$ suffice and these recursive calls are ignored for the moment).
 Let $\varphi'\colon H' \rightarrow S_{k'}$ be the giant
 representation computed on the next level of the recursion where $H'$ is the projection of $H''$ to an invariant subset of the domain for some $H'' \leq H$
 (if no giant representation is computed then the algorithm performs standard Luks reduction and the node on the next level is a leaf).
 Observe that $|(H')^{\mathfrak{C}'}| \geq \frac{(k')!}{2}$ because $(H')^{\varphi'} \geq A_{k'}$ and $H'_{(\mathfrak{C}')} \leq \ker(\varphi')$.
 Also note that $|H^{\mathfrak{C}'}| \leq \frac{k!}{(4/3)^{k}}$ since $\ker(\varphi) = G'_{(\mathfrak{C}')}$ by Lemma \ref{la:compute-giant-representation}.
 So \[\frac{(k')!}{2} \leq \frac{k!}{(4/3)^{k}}.\]
 Hence,
 \[(4/3)^{k} \leq 2 \cdot 2^{(k - k') \log k} \leq (4/3)^{3(k-k')\log k}\]
 since $k$ is sufficiently large.
 So \[k' \leq k - \frac{k}{3 \log k}.\]
 It follows that the height of the recursion tree is $\mathcal{O}((\log d)^{2})$.
 Thus, the number of nodes of the recursion tree is bounded by $d^{\mathcal{O}((\log d)^{2})} = 2^{\mathcal{O}((\log d)^{3})}$.
 By Lemma \ref{la:aggregate-local-certificates}, \ref{la:find-structure} and \ref{la:find-symmetry} each node of the recursion tree makes recursive calls to String Isomorphism over domain sizes $n_i \leq n/2$
 where $\sum_i n_i \leq 2^{\mathcal{O}((\log d)^{2})}n$ and uses additional computation $d^{\mathcal{O}((\log d)^{c})}n^{c}$ for some constant $c$.
 Putting this together, the desired bound follows.
\end{proof}

\begin{theorem}
 \label{thm:string-isomorphism-almost-d-ary}
 Let $G \leq \Sym(\Omega)$ be a permutation group and let $\mathfrak{x},\mathfrak{y} \colon \Omega \rightarrow \Sigma$ be two strings.
 Moreover, suppose there is an almost $d$-ary sequence of $G$-invariant partitions
 $\{\Omega\} = \mathfrak{B}_0 \succ \dots \succ \mathfrak{B}_m = \{\{\alpha\} \mid \alpha \in \Omega\}$.
 Then one can compute $\Iso_G(\mathfrak{x},\mathfrak{y})$ in time
 $n^{\mathcal{O}((\log d)^{c})}$, for an absolute constant $c$.
\end{theorem}

\begin{algorithm}
 \caption{String Isomorphism}
 \label{alg:string-isomorphism-alg}
 \DontPrintSemicolon
 \SetKwInOut{Input}{Input}
 \SetKwInOut{Output}{Output}
 \Input{$G \leq \Sym(\Omega)$ a $\ourgamma_d$-group, $\mathfrak{x},\mathfrak{y}\colon \Omega \rightarrow \Sigma$ two strings
        and an almost $d$-ary sequence of $G$-invariant partitions
        $\{\Omega\} = \mathfrak{B}_0 \succ \dots \succ \mathfrak{B}_m = \{\{\alpha\} \mid \alpha \in \Omega\}$.}
 \Output{$\Iso_G(\mathfrak{x},\mathfrak{y})$}
 \BlankLine
 \eIf{$G$ is not transitive}{
  compute orbits $\Omega_1,\dots,\Omega_s$\;
  recursively process group orbit by orbit \tcc*[r]{\small restrict partitions to orbits}
  \Return $\Iso_G(\mathfrak{x},\mathfrak{y})$\;
 }{
  \eIf{$G^{\mathfrak{B}_1}$ is semi-regular}{
   apply standard Luks reduction \tcc*[r]{\small restrict partitions to orbits of $G_{(\mathfrak{B}_1)}$}
   \Return $\Iso_G(\mathfrak{x},\mathfrak{y})$\;
  }(\tcc*[f]{\small assumptions of Lemma~\ref{lem:partition-small-size-recursion} are satisfied}){
   apply Lemma~\ref{lem:partition-small-size-recursion}\;
   \Return $\Iso_G(\mathfrak{x},\mathfrak{y})$\;
  }
 }
\end{algorithm}

\begin{proof}
 Consider Algorithm \ref{alg:string-isomorphism-alg}.
 The algorithm essentially distinguishes between two cases. 
 If the input group $G$ is not transitive or the action of $G$ on the block system $\mathfrak{B}_1$ is semi-regular, the algorithm follows Luks algorithm recursively computing the set $\Iso_G(\mathfrak{x},\mathfrak{y})$.
 In the other case $G$ is transitive and $|\mathfrak{B}_1| \leq d$ and hence, we can apply Lemma~\ref{lem:partition-small-size-recursion} to recursively compute $\Iso_G(\mathfrak{x},\mathfrak{y})$.
 
 Clearly, it computes the desired set of isomorphisms.
 The bound on the running follows from Lemma \ref{la:recursion-bound}.
 Note that the bottleneck is the type of recursion used in Lemma~\ref{lem:partition-small-size-recursion}.
 Also observe that every group $H$, for which the algorithm performs a recursive call, is the projection of a subgroup of $G$ to an invariant subset of the domain.
 Hence, by restricting the partitions $\mathfrak{B}_0,\dots,\mathfrak{B}_m$ to the domain of $H$
 one obtains a sequence of partitions for the group $H$ with the desired properties (cf.\ Observation \ref{obs:sequence-of-partitions}).
\end{proof}

Combining Theorem \ref{thm:reduction-two} and Theorem~\ref{thm:string-isomorphism-almost-d-ary} we obtain the main technical result of this work.

\begin{theorem}
 \label{thm:main-result-gamma-d}
 Let $G \leq \Sym(\Omega)$ be a $\ourgamma_d$-group and let $\mathfrak{x},\mathfrak{y} \colon \Omega \rightarrow \Sigma$ be two strings.
 Then there is an algorithm deciding whether $\mathfrak{x} \cong_G \mathfrak{y}$ in time $n^{\mathcal{O}((\log d)^{c})}$, for an absolute constant $c$.
\end{theorem}

\begin{proof}
 Using orbit-by-orbit processing we can assume that the group $G$ is transitive.
 For a transitive group the statement follows by first applying Theorem \ref{thm:reduction-two} and then Theorem~\ref{thm:string-isomorphism-almost-d-ary}.
\end{proof}

\section{Applications}\label{sec:applications}

\subsection{Isomorphism for graphs of bounded degree}

Using the improved algorithm for string isomorphism we can now prove the main result of this work using the following well-known reduction.

\begin{theorem}[\cite{luks82,BL83}]
 \label{thm:reduction-gi-to-si}
 There is a polynomial-time Turing-reduction from the Graph Isomorphism Problem for graphs of maximum degree $d$ to the String Isomorphism Problem for $\ourgamma_d$-groups
 (the running time of the reduction does not depend on $d$).
\end{theorem}

The reduction follows \cite{luks82} using an additional trick presented in \cite[Section 4.2]{BL83} to remove the dependence of the running time on $d$.

Combining this reduction with the improved algorithm for string isomorphism, we get the desired algorithm for isomorphism tests of bounded degree graphs.

\begin{theorem}[Theorem~\ref{thm:main-result-degree-d} restated]
 The Graph Isomorphism Problem for graphs of maximum degree $d$ can be solved in time $n^{\mathcal{O}((\log d)^{c})}$, for an absolute constant $c$.
\end{theorem}

\begin{proof}
 This follows from Theorem \ref{thm:main-result-gamma-d} and \ref{thm:reduction-gi-to-si}.
\end{proof}

\subsection{Isomorphism for relational structures and hypergraphs}

For the second application of Theorem~\ref{thm:main-result-gamma-d} consider the isomorphism problem for relational structures.

\begin{theorem}
  \label{thm:relational}
 Let $\mathfrak{A} = (D,R)$, $\mathfrak{A}' = (D,R')$ be relational structures where $R,R' \subseteq D^{t}$ are $t$-ary relations.
 Then one can decide whether $\mathfrak{A}$ is isomorphic to $\mathfrak{A}'$ in time $n^{\mathcal{O}(t \cdot (\log n)^{c})}$ where $n = |D|$.
\end{theorem}

\begin{proof}
 Let $\mathfrak{x}\colon D^{t} \rightarrow \{0,1\}$ be the string with $\mathfrak{x}(a_1,\dots,a_t) = 1$ if and only if $(a_1,\dots,a_t) \in R$.
 Similarly define the string $\mathfrak{x}'\colon D^{t} \rightarrow \{0,1\}$ for the relation $R'$.
 Now let $G = \Sym(D)^{(D^{t})}$ be the symmetric group over the set $D$ with its natural action on $t$-tuples.
 Then $\mathfrak{A}$ is isomorphic to $\mathfrak{A}'$ if and only if $\mathfrak{x}$ is $G$-isomorphic to $\mathfrak{x}'$.
 Moreover, $G \in \ourgamma_n$.
 Hence, by Theorem~\ref{thm:main-result-gamma-d}, one can decide in time $n^{\mathcal{O}(t \cdot (\log n)^{c})}$ whether $\mathfrak{x}$ is $G$-isomorphic to $\mathfrak{x}'$.
\end{proof}

In many cases this leads to a better running time than first translating the structure into a graph and than applying Babai's algorithm to test whether the two resulting graphs are isomorphic.
In particular, in case the arity $t$ is large and also the size of the relation is large our method gives a much better worst case complexity than the other approach.

Also note that as a special case the same running time can be obtained for hypergraphs if $t$ is the maximal hyperedge size.
This also improves on previous results (see e.g.\ \cite{BC08}).

\begin{corollary}\label{cor:hypergraph}
 Let $\mathcal{H} = (V,\mathcal{E})$, $\mathcal{H}' = (V,\mathcal{E}')$ be two hypergraphs such that every hyperedge $E \in \mathcal{E} \cup \mathcal{E}'$ has size $|E| \leq t$.
 Then one can decide whether $\mathcal{H}$ is isomorphic to $\mathcal{H}'$ in time $n^{\mathcal{O}(t \cdot (\log n)^{c})}$ where $n = |V|$.
\end{corollary}

\section{Concluding Remarks}
We have obtained a new graph isomorphism test with a running time
bounded by a polynomial of degree polylogarithmic in the
maximum degree of the input graphs. Technically, this result relies on
some heavy group theory, new combinatorial tricks that allow us to
reduce the string isomorphism problem for $\ourgamma_d$ groups to
a setting where we have an ``almost $d$-ary'' sequence of invariant partitions
controlling the operation of the groups, and a refinement of the
techniques introduced by Babai~\cite{Babai16} for his quasipolynomial time
isomorphism test.

We hope that the machinery we have developed here will have further
applications and ultimately even lead to an improvement of Babai's
isomorphism test. More immediate applications may be obtained for the
isomorphism problem under restrictions of other parameters than the
maximum degree. For example, we conjecture that there also is an
isomorphism test running in time $n^{\mathcal{O}((\log k)^c)}$,
where $k$ is the tree width of the input graphs.
We remark that the results established in this work have already been used in \cite{GroheNSW18}
to obtain an improved fpt algorithm for isomorphism parameterized by tree width.

Another related problem that we leave open is whether the graph
isomorphism problem parameterized by the maximum degree of the input
graphs is fixed-parameter tractable.

\bibliographystyle{plain}
\bibliography{literature}

\begin{appendix}

\section{Aggregating Local Certificates}
\label{app:proof-aggregate}

We give a proof for Lemma \ref{la:aggregate-local-certificates}.

\begin{definition}
 A group $G \leq \Sym(\Omega)$ is \emph{$t$-transitive} if its natural induced action on the set of $n(n-1)\dots(n-t+1)$ ordered $t$-tuples of distinct elements is transitive.
 The \emph{degree of transitivity} $d(G)$ is the largest $t$ such that $G$ is $t$-transitive.
\end{definition}

\begin{theorem}[CFSG]
 \label{thm:degree-of-transitivity}
 Let $G \leq \Sym(\Omega)$ be a non-giant group.
 Then $d(G) \leq 5$.
\end{theorem}

A slightly weaker statement, namely $d(G) \leq 7$ for all non-giants permutation groups, can be shown using only Schreier's Hypothesis (see \cite[Theorem 7.3A]{DM96}).

\begin{lemma}[cf.\;\cite{Babai15-full}, Corollary 2.4.13]
 \label{la:symmetry-defect-regular-graph}
 Let $\Gamma$ be a non-trivial regular graph. Then the relative symmetry defect of $\Gamma$ is at least $1/2$. 
\end{lemma}

\begin{lemma}[Lemma \ref{la:aggregate-local-certificates} restated]
 Let $\mathfrak{x}_1,\mathfrak{x}_2\colon \Omega \rightarrow \Sigma$ be two strings,
 $G \leq \Sym(\Omega)$ be a group and suppose there is an almost $d$-ary sequence of $G$-invariant partitions
 $\{\Omega\} = \mathfrak{B}_0 \succ \dots \succ \mathfrak{B}_m = \{\{\alpha\} \mid \alpha \in \Omega\}$.
 Furthermore suppose there is a giant representation $\varphi\colon G \rightarrow S_k$.
 Let $\max\{8,2 + \log_2d\} < t < k/10$.
 
 Then there are natural numbers $\ell \in \mathbb{N}$ and $n_1,\dots,n_\ell \leq n/2$ such that $\sum_{i=1}^{\ell}n_i \leq k^{\mathcal{O}(t)}n$ and,
 for each $i \in [\ell]$ using a recursive call to String Isomorphism
 over domain size $n_i$, and $k^{\mathcal{O}(t)}n^{c}$
 additional computation,
 one obtains for $i=1,2$ one of the following:
 \begin{enumerate}
  \item a family of $r\le k^6$ many $t$-ary relational structures $\mathfrak{A}_{i,j}$, for $j\in[r]$, associated with $\mathfrak x_i$,
    each with domain $D_{i,j}\subseteq[k]$ of size $|D_{i,j}| \geq \frac{3}{4}k$ and with relative symmetry defect at least $\frac{1}{4}$ such that
    \[\left\{\mathfrak{A}_{1,1},\dots,\mathfrak{A}_{1,r}\right\}^{\varphi(g)} = \left\{\mathfrak{A}_{2,1},\dots,\mathfrak{A}_{2,r}\right\} \text{ for every } g \in \Iso_G(\mathfrak{x}_1,\mathfrak{x}_2),\]
    or
  \item a subset
    $\Delta_i \subseteq [k]$ associated with $\mathfrak{x}_i$ of size
    $|\Delta_i| \geq \frac{3}{4}k$ and $K_i \leq \Aut_{G_{\Delta_i}}(\mathfrak{x}_i)$
    such that $(K_i^{\varphi})^{\Delta_i} \geq \Alt(\Delta_i)$ and
    \[\Delta_1^{\varphi(g)} = \Delta_2 \text{ for every } g \in \Iso_G(\mathfrak{x}_1,\mathfrak{x}_2).\]
 \end{enumerate}
\end{lemma}

\begin{proof}
 For every $t$-element subset $T \subseteq [k]$ determine whether $T$ is full (with respect to $\mathfrak{x}_i$) and compute a corresponding certificate using Lemma \ref{la:local-certificate}.
 Let $F_i \leq \Sym(\Omega)$ be the group generated by the fullness-certificates for all full subsets $T \subseteq [k]$ with respect to string $\mathfrak{x}_i$.
 Note that the group $F_i$ is canonical. Let $S_i \subseteq [k]$ be the support of $F_i^{\varphi}$ (the set of elements not fixed by $F_i^{\varphi}$).
 
 First suppose $\frac{1}{4}k \leq |S_i| \leq \frac{3}{4}k$.
 Then one obtains a canonical structure $\mathfrak{A}_i$ with domain $[k]$ and relative symmetry defect at least $\frac{1}{4}$ by coloring each element $\alpha \in [k]$ depending on whether $\alpha \in S_i$.
 
 Next suppose $|S_i| > \frac{3}{4}k$.
 We distinguish between three subcases.
 First assume there is no orbit of $F_i^{\varphi}$ of size at least $\frac{3}{4}k$.
 Then the partition into the orbits of $F_i^{\varphi}$ gives a canonical structure $\mathfrak{A}_i$ with domain $[k]$ and relative symmetry defect at least $\frac{1}{4}$.
 So assume there is a (unique) orbit $C \subseteq [k]$ of size $C \geq \frac{3}{4}k$.
 If $(F_i^{\varphi})^{C} \geq \Alt(C)$ then the second option of the Lemma is satisfied.
 Hence suppose $(F_i^{\varphi})^{C}$ is not a giant.
 By Theorem \ref{thm:degree-of-transitivity} the degree of transitivity satisfies $d((F_i^{\varphi})^{C}) \leq 5$.
 Let $I \subseteq C$ be an arbitrary set of size $d((F_i^{\varphi})^{C}) - 1$ and individualize the elements of $I$.
 Then $(F_i^{\varphi})_{(I)}^{C'}$ is transitive, but not $2$-transitive, where $C' = C \setminus I$.
 Note that the number of possible choices for the set $I$ is at most $k^{4}$.
 Now let $\mathfrak{X}_i = (C',R_1,\dots,R_r)$ be the \emph{orbital configuration} of $(F_i^{\varphi})_{(I)}$ on the set $C'$, that is, the relations $R_i$ are the orbits of $(F_i^{\varphi})_{(I)}$ in its natural action on $C' \times C'$.
 Note that $r \geq 3$ since $(F_i^{\varphi})_{(I)}^{C'}$ is not $2$-transitive.
 Also observe that the numbering of the $R_i$ is not canonical (isomorphisms may permute the $R_i$).
 Without loss of generality suppose that $R_1$ is the diagonal.
 Now individualize one of the $R_i$ for $i \geq 2$ at a multiplicative cost of $r-1 \leq k-1$.
 If $R_i$ is undirected (i.e.\;$R_i = R_i^{-1}$) then it defines a non-trivial regular graph.
 Since the symmetry defect of this graph is at least $1/2$ (see Lemma \ref{la:symmetry-defect-regular-graph}) this gives us the desired structure.
 Otherwise $R_i$ is directed.
 If the out-degree of a vertex is strictly less $(|C'|-1)/2$ then the undirected graph $\Gamma = (C',R_i \cup R_i^{-1})$ is again a non-trivial regular graph.
 Otherwise, by individualizing one vertex (at a multiplicative cost of $|C'| \leq k$), one obtains a coloring of symmetry defect at least $1/2$ by coloring vertices depending on whether they are an in- or out-neighbor of the individualized vertex.
 
 Finally suppose $|S_i| < \frac{1}{4}k$.
 Let $D_i = [k] \setminus S_i$.
 Indeed it can be assumed that $|D_1| = |D_2| \geq \frac{3}{4}k$.
 Observe that every $T \subseteq D_i$ is not full with respect to the string $\mathfrak{x}_i$.
 Let $D_i' = D_i \times \{i\}$ (to make the sets disjoint).
 
 Consider the following category $\mathcal{L}$.
 The objects are the pairs $(T,i)$ where $T \subseteq D_i$ is a
 $t$-element subset.
 The morphisms $(T,i) \rightarrow (T',i')$ are the bijections computed
 in Lemma \ref{la:compare-local-certificates} for the test sets $T$
 and $T'$ along with the corresponding strings.
 The morphisms define an equivalence relation on the set
 $(D_1')^{\angles{t}} \cup (D_2')^{\angles{t}}$
 where $(D_i')^{\angles{t}}$ denotes the set of all ordered $t$-tuples with distinct elements over the set $D_i'$.
 Let $R_1,\dots,R_r$ be the equivalence classes and define $R_j(i) = R_j \cap (D_i')^{\angles{t}}$.
 Then $\mathfrak{A}_i = (D_i',R_1(i),\dots,R_r(i))$ is a canonical $t$-ary relational structure.
 Moreover, the symmetry defect of $\mathfrak{A}_i$ is at least $|D_i| - t +1 \geq |D_i|/4$.
\end{proof}

\end{appendix}

\end{document}